\documentclass[11pt]{article}

\usepackage{a4wide}
\usepackage{fncylab,enumerate,wrapfig,placeins}

\usepackage{graphicx} 
\usepackage{amsmath,amsthm,amsfonts,amssymb}  
\usepackage{etoolbox}% http://ctan.org/pkg/etoolbox
\usepackage[usenames,dvipsnames]{color}
\usepackage{color}

 \usepackage[colorlinks,%
linkcolor=BrickRed,% 
filecolor =BrickRed,%
citecolor=RoyalPurple,%
]{hyperref}

\usepackage{figlatex}

\def\mindex#1{\index{#1}}

\newcommand{\bbblot}{\raise1pt\hbox{\vrule height .4ex width .4ex depth .05ex}}
\long\def\defbox#1{\framebox[.9\hsize][c]{\parbox{.85\hsize}{%
\parindent=0pt
\baselineskip=12pt plus .1pt      % STYLE 
\parskip=6pt plus 1.5pt minus 1pt % CHANGES
 #1}}}

\long\def\beginbox#1\endbox{\subsection*{}%
\hbox{\hspace{.05\hsize}\defbox{\medskip#1\bigskip}}%
\subsection*{}}

\def\endbox{}

 \def\archival#1{} %Notes I'd like to save
 
 \def\notes#1{\typeout{check notes!!!}}   %  For final version

\def\FRAC#1#2#3{\genfrac{}{}{}{#1}{#2}{#3}}

\def\ddtp{{\mathchoice{\FRAC{1}{d^{\hbox to 2pt{\rm\tiny +\hss}}}{dt}}%
{\FRAC{1}{d^{\hbox to 2pt{\rm\tiny +\hss}}}{dt}}%
{\FRAC{3}{d^{\hbox to 2pt{\rm\tiny +\hss}}}{dt}}%
{\FRAC{3}{d^{\hbox to 2pt{\rm\tiny +\hss}}}{dt}}}}

\def\ddyp{{\mathchoice{\FRAC{1}{d^{\hbox to 2pt{\rm\tiny +\hss}}}{dy}}%
{\FRAC{1}{d^{\hbox to 2pt{\rm\tiny +\hss}}}{dy}}%
{\FRAC{3}{d^{\hbox to 2pt{\rm\tiny +\hss}}}{dy}}%
{\FRAC{3}{d^{\hbox to 2pt{\rm\tiny +\hss}}}{dy}}}}

\def\half{{\mathchoice{\FRAC{1}{1}{2}}%
{\FRAC{1}{1}{2}}%
{\FRAC{3}{1}{2}}%
{\FRAC{3}{1}{2}}}}

\def\sign{{\rm sign}}

\newsavebox{\junk}
\savebox{\junk}[1.6mm]{\hbox{$|\!|\!|$}}

\def\limsup{\mathop{\rm lim{\,}sup}}

\def\argmin{\mathop{\rm arg{\,}min}}
\def\argmax{\mathop{\rm arg{\,}max}}

\def\dstate{{\sf X_\diamond}}

\def\bfmath#1{{\mathchoice{\mbox{\boldmath$#1$}}%
{\mbox{\boldmath$#1$}}%
{\mbox{\boldmath$\scriptstyle#1$}}%
{\mbox{\boldmath$\scriptscriptstyle#1$}}}}

\def\bfmhaI{\bfmath{\hat I}}

\def\bfDelta{\bfmath{\Delta}}

\def\bfmx{\bfmath{x}}

\def\bfmA{\bfmath{A}}

\def\bfmhaI{\bfmath{\haI}}

\def\bfmQ{\bfmath{Q}}

\def\bfmU{\bfmath{U}}

\def\bfmX{\bfmath{X}}

\def\bfmY{\bfmath{Y}}

\def\bfmhhaY{\bfmath{\hhaY}} %\widehat{\widehat{Y}}}}
\def\bfmhhaY{\hbox to 0pt{$\widehat{\bfmY}$\hss}\widehat{\phantom{\raise 1.25pt\hbox{$\bfmY$}}}}

\def\bfmhaW{\bfmath{\haW}}

\def\bfmhaW{\bfmath{\widehat W}}

\def\hah{{\hat h}}

\def\tilQ{{\widetilde Q}}

\def\tilW{\widetilde W}
\def\tilX{\widetilde X}

\def\tilx{\tilde x}

\def\clE{{\cal E}}

\def\clG{{\cal G}}

\def\clN{{\cal N}}

\def\eqdef{\mathbin{:=}}

\def\Prob{{\sf P}}

\def\Expect{{\sf E}}

\def\lgmath#1{{\mathchoice{\mbox{\large #1}}%
{\mbox{\large #1}}%
{\mbox{\tiny #1}}%
{\mbox{\tiny #1}}}}

\def\One{{\mathchoice{\lgmath{\sf 1}}%
{\mbox{\sf 1}}%
{\mbox{\tiny \sf 1}}%
{\mbox{\tiny \sf 1}}}}

 \def\eq#1/{(\ref{#1})}

\def\ind{\bbbone}
  
 \def\epsy{\varepsilon}

\theoremstyle{remark}

\newtheorem{theorem}{Theorem}[section]
\newtheorem{lemma}[theorem]{Lemma}
\newtheorem{proposition}[theorem]{Proposition}

\def\Lemma#1{Lemma~\ref{#1}}
\def\Proposition#1{Proposition~\ref{#1}}
\def\Theorem#1{Theorem~\ref{#1}}
 
\def\Section#1{Section~\ref{#1}}

\def\barc{{\overline {c}}}

\def\barg{{\overline {g}}}

\def\barn{{\overline {n}}}

\def\barq{{\overline {q}}}
\def\barr{{\overline {r}}}

\def\bareta{{\overline{\eta}}}

{\end{list}}

\def\ass(#1:#2){(#1\ref{#1:#2})}

\def\ritem#1{
\item[{\sf \ass(\current_model:#1)}]
}

\newenvironment{recall-ass}[1]{% 
\begin{description}
\def\current_model{#1}}{
\end{description}
}

\def\sq{\hbox{\rlap{$\sqcap$}$\sqcup$}}
\def\qed{\ifmmode\sq\else{\unskip\nobreak\hfil
\penalty50\hskip1em\null\nobreak\hfil\sq
\parfillskip=0pt\finalhyphendemerits=0\endgraf}\fi}

\newcommand{\blot}{\vrule height 1.1ex width .9ex depth -.1ex }
\def\qedb{\ifmmode\blot\else{\vspace{-.2cm}\unskip\nobreak\hfil
\penalty50\hskip1em\null\nobreak\hfil\blot
\parfillskip=0pt\finalhyphendemerits=0\endgraf}\fi}

\newcounter{l1}
\newcounter{l2}
\newcounter{l3}

\newcommand{\barablist}{\begin{list}{\arabic{l1}}{\usecounter{l1}}}
\newcommand{\balphlist}{\begin{list}{(\alph{l2})}{\usecounter{l2}} 
\setlength{\topsep}{0pt}
\setlength{\partopsep}{0pt}
\setlength{\itemsep}{0pt}
\setlength{\parsep}{5pt}
\setlength{\leftmargin}{-10pt}
\setlength{\rightmargin}{0pt}
\setlength{\itemindent}{-5pt}
}
\newcommand{\bromalist}{\begin{list}{(\roman{l3})}{\usecounter{l3}}}

\newcounter{rmnum}
\newenvironment{romannum}{\begin{list}{{\upshape (\roman{rmnum})}}{\usecounter{rmnum}
\setlength{\leftmargin}{14pt}
\setlength{\rightmargin}{12pt}
\setlength{\itemindent}{1pt}
}}{\end{list}}

\newcounter{anum}

\def\eq#1/{(\ref{e:#1})}

\newcommand{\field}[1]{\mathbb{#1}}

\def\Re{\field{R}}

\def\nat{\field{Z}_+}

\def\Prob{{\sf P}}

\def\Expect{{\sf E}}

\def\dU{{\sf U_\diamond}}

\def\dA{{\sf A_\diamond}}
\def\dQ{{\sf Q_\diamond}}

\def\R{{\sf R}}

 \def\transpose{{\hbox{\rm\tiny T}}}

\def\sign{{\rm sign}}

\def\argmin{\mathop{\rm arg\, min}}
\def\ind{\hbox{\large \bf 1}}

\def\epsy{\varepsilon}

\def\tilw{\widetilde{w}}

\def\haW{\widehat{W}}
\def\haY{\widehat{Y}}

\def\hhaY{\hbox to 0pt{$\haY$\hss}\widehat{\phantom{\raise 1.25pt\hbox{Y}}}}

\def\haI{{\hat I}}

\def\haeta{\widehat\eta}

  \def\blackqed{\ifmmode\text{$\blacksquare$}\else{\unskip\nobreak\hfil
\penalty50\hskip1em\null\nobreak\hfil\text{$\blacksquare$}
\parfillskip=0pt\finalhyphendemerits=0\endgraf}\fi}

\newlength{\noteWidth}
\long\def\notes#1{\ifinner
           {\tiny #1}
           \else
           \marginpar{\parbox[t]{\noteWidth}{\raggedright\tiny #1}}
       \fi\typeout{#1}}

\def\notes#1{\typeout{read notes: #1}}  %uncomment for final version

\def\wham#1{\smallbreak\noindent\textbf{#1} \ }
\def\yep#1{\typeout{Answered!!}}

\def\spm#1{\notes{spm: #1}}

\graphicspath{{figures/}}

\def\Ebox#1#2{%
\begin{center}    
\includegraphics[width=#1\hsize]{#2}
\end{center}} 

\def\Fig#1{Fig.~\ref{#1}}

\def\Prop#1{Prop.~\ref{#1}}     
\def\Lemma#1{Lemma~\ref{#1}}  
\def\Thm#1{Theorem~\ref{#1}}     

\newdimen\slantmathcorr
\def\oversl#1{%assuming that mathslant=0.25
\setbox0=\hbox{$#1$}
\slantmathcorr=\wd0
\hskip 0.2\slantmathcorr \overline{\hbox to 0.8\wd0{%
\vphantom{\hbox{$#1$}}}}
\hskip-\wd0\hbox{$#1$}
}

\def\undersl#1{%assuming that mathslant=0.25
\setbox0=\hbox{$#1$}
\slantmathcorr=\wd0
\underline{\hbox to 0.8\wd0{%
\vphantom{\hbox{$#1$}}}}
\hskip-0.8\wd0\hbox{$#1$}
}

\def\maxdelta{\bar{\delta}_\bullet}
\def\bardelta{\bar{\delta}}

\def\udelta{\undersl{\delta}} 

\def\barIdle{{\Large\text{$\iota$}}}
 
\def\cIDP{c_{\text{{\sf\tiny idp}}}}
\def\minc{c_{\text{\it gap}}} 

\def\pcm{p_{I}}

\def\cmax{c_{\text{\it max}}}

\def\barcpD{\barc_+^\dmdD}
\def\barcmD{\barc_-^\dmdD}
\def\barcpS{\barc_+^\splyS}
\def\barcmS{\barc_-^\splyS}

\def\thresh{\tau^*}
\def\threshW{\tau^\bullet}

 \def\Dmd{{\cal D}}
 \def\Sply{{\cal S}} 
 
\def\IDmd{{\cal D}}
 \def\ISply{{\cal S}}
 \def\IEdge{{\cal E}}
 \def\IArrive{{\cal A}}

\def\dmdD{{\hbox{\tiny\raisebox{-.03cm}{D}}}}
\def\splyS{{\hbox{\tiny\raisebox{-.03cm}{S}}}}

\def\dmd{{\hbox{\sf\tiny\raisebox{-.03cm}{$\cal D$}}}}
\def\sply{{\hbox{\sf\tiny\raisebox{-.03cm}{$\cal S$}}}}

\def\resG{\clG_{D \cup S}}
\def\resEdge{\IEdge_{D \cup S}}

\def\NN{$\text{\sf N}\hspace{-.092cm}\text{\sf N}$}

\def\umax{\barn_u}

\begin{document}

\renewcommand*{\thefootnote}{\fnsymbol{footnote}}

%\begin{frontmatter}

\title{Approximate optimality with bounded regret\\
in dynamic matching models}
\author{Ana Bu\v{s}i\'c\thanks{Inria and the Computer Science Dept.\ of \'Ecole Normale Sup\'erieure, Paris, France. E-mail: {\tt ana.busic@inria.fr}.} \and 
Sean Meyn\thanks{Department of Electrical and Computer Engg.\ at the University of Florida, Gainesville. E-mail: {\tt meyn@ece.ufl.edu}.}
}
\date{\empty}
\maketitle

%\title{Approximate optimality with bounded regret in dynamic matching models}
%\runtitle{Approximate optimality  in dynamic matching models}

%\begin{aug}
%\author{\fnms{Ana} \snm{Bu\v{s}i\'c}\thanksref{t2,m1}\ead[label=e1]{ana.busic@inria.fr}},
%\and
%\author{\fnms{Sean} \snm{Meyn}\thanksref{t3,m2}\ead[label=e2]{meyn@ece.ufl.edu}}
%
%
%
%
%%\thankstext{t1}{Some comment}
%\thankstext{t2}{Support from  the French National Research Agency under award ANR-12-MONU-0019 is gratefully acknowledged}
%\thankstext{t3}{Support from the National Science Foundation is gratefully acknowledged, under NSF awards CPS-0931416 and CPS-1259040}
%
%\runauthor{A. Bu\v{s}i\'c and  S. Meyn}
%
%\affiliation{Inria\thanksmark{m1} and University of Florida\thanksmark{m2}}
%
%\address{Inria and the Computer Science Dept.\ of \'Ecole Normale Sup\'erieure, 
%Paris, France
%\\ \textit{and}  Department of Electrical and Computer Engg.\ at the University of Florida, Gainesville\\
%\printead{e1}\\
%\phantom{E-mail:\ }\printead*{e2}}
%
% 
%\end{aug}

%\noindent\textbf{Abstract}
\begin{abstract}
%We consider a dynamic bipartite matching model with random arrivals.   In prior work, authors have proposed policies that are stabilizing,  and also policies that are approximately finite-horizon optimal.  This paper considers the infinite-horizon average-cost optimal control problem.
%%
%A relaxation of the stochastic control problem is proposed,  which is found to be a special case of an inventory model, as treated in the classical theory of Clark and Scarf.    The optimal policy for the relaxation admits a closed-form expression.   Based on the policy for this relaxation,  a new matching policy is proposed.    For a parameterized family of models in which the network load approaches capacity,   this policy is shown to be approximately optimal, with bounded regret,  even though the average cost grows without bound.
We consider a discrete-time bipartite matching model with random arrivals of units of \emph{supply} and \emph{demand} that can wait in queues located at the nodes in the network. A control policy determines which are matched at each time. The focus is on the infinite-horizon average-cost optimal control problem.
A relaxation of the stochastic control problem is proposed, which is found to be a special case of an inventory model, as treated in the classical theory of Clark and Scarf. The optimal policy for the relaxation admits a closed-form expression. Based on the policy for this relaxation, a new matching policy is proposed. For a parameterized family of models in which the network load approaches capacity, this policy is shown to be approximately optimal, with bounded regret, even though the average cost grows without bound. 
 
\end{abstract}

%\begin{keyword}[class=MSC]
%\kwd[Primary ]{90B15}
%\kwd{Network models, stochastic}
%\kwd[; secondary ]{93E20}
%\kwd{Optimal stochastic control}
%\end{keyword}

%\begin{keyword}
%\kwd{Matching models}
%\kwd{Stochastic optimal control}
%\kwd{Queueing networks}
%\end{keyword}


\section{Introduction}
\label{s:intro}

%\begin{wrapfigure}{l}{.35\hsize}
%\vspace{-.15cm}
%\Ebox{.98}{NNv2.pdf}
%\vspace{-.25cm}\caption{\small
%The \NN\ network}
%\label{f:NN} 
%\vspace{-.1cm}
%\end{wrapfigure} 

%
%We consider a dynamic matching model with random arrivals -- a   stochastic version of the bipartite matching model. 
%As in the static setting, it is based on a bipartite graph --- a simple example is shown in \Fig{f:NN}.  In the discrete-time dynamic model there are arrivals of units of `supply' and `demand'   that can wait in queues located at the nodes in the network.  A control policy determines which are matched at each time.
%% --- the \textit{stable marriage} model of   Gale and   Shapley.    
%%In recent work,  Bu\v{s}i\'c et.~al.\ in \cite{busgupmai13} considered a similar model, where they established necessary conditions for the existence of a stabilizing policy, and established examples of stabilizing policies.

We consider a discrete-time bipartite matching model with random arrivals of units of `supply' and `demand' that can wait in queues located at the nodes in the network. A control policy determines which are matched at each time. We focus on the infinite-horizon average-cost optimal control problem.
A relaxation of the stochastic control problem is proposed, which is found to be a special case of an inventory model, as treated in the classical theory of Clark and Scarf. The optimal policy for the relaxation admits a closed-form expression. Based on the policy for this relaxation, a new matching policy is proposed. For a parameterized family of models in which the network load approaches capacity, this policy is shown to be approximately optimal, with bounded regret, even though the average cost grows without bound.

\begin{figure}[h]
   \centering
   \includegraphics[width=.6\textwidth]{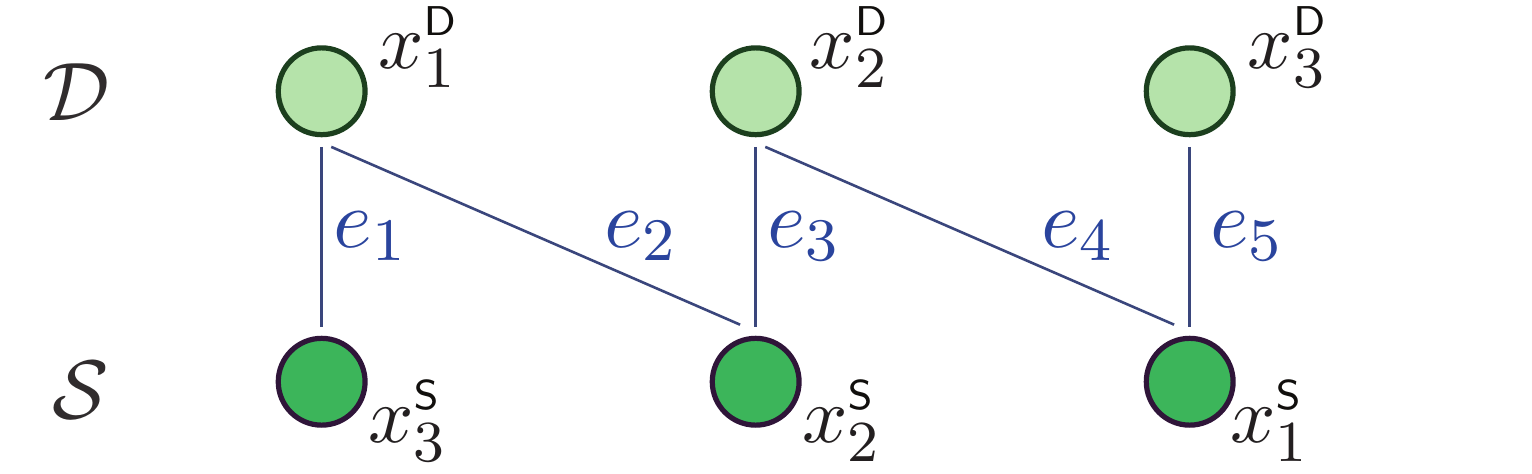}
  \caption{\small
The \NN\ network.}
\label{f:NN} 
 \end{figure}

The theory of matching has a long history in economics, mathematics, and graph theory \cite{dubfree81,lovplu09}, with
applications  found in many other areas such as  chemistry and information theory.  Most of the work is in a static setting.  
The dynamic model has received recent attention in \cite{busgupmai13,gurwar14}.

The most compelling application   is organ donation:
United Network for Organ Sharing (UNOS) offers kidney paired donation (KPD).  This is a transplant option for candidates who have a living donor who is medically able, but cannot donate a kidney to their intended candidate because they are incompatible (i.e., poorly matched)  \cite{KidneyPaired}.
 In this application, or  application
to resource allocation (such as in
scheduling in a power grid) \cite{beronn08,zdedecche09}, %communication networks
%\cite{han10}, 
or pattern recognition \cite{sch99}, data
arrives sequentially and randomly, so that matching decisions must be made in
real-time,  taking into account the uncertainty of future requirements for
supply or demand,  or the uncertainty of the sequence of classification tasks to
be undertaken.  
The choice of matching decisions can be cast as an optimal control problem for a dynamic matching model. 

This paper builds upon the prior work \cite{busgupmai13} that established necessary and sufficient conditions for stability of a dynamic matching model (in the sense that there exists a Markovian matching policy for which the controlled process is positive Harris recurrent),    and gave several examples of policies that have maximal stability region (sometimes known as ``throughput optimal''). 
The goal in the present work is to obtain a better understanding of the structure of optimal policies.   Based on this, we seek policies with good performance,  as quantified by average-cost of the Markovian model:  given a linear cost function $c$ on buffer levels, the average cost is the long-run average, 
\begin{equation}
\eta = \limsup_{N\to\infty} \frac{1}{N}\sum_{t=0}^{N-1} \Expect\bigl[  c(X(t)) \bigr]
\label{e:eta}
\end{equation}
This will in general depend on the initial conditions,  and on the policy that determines matching decisions.

These goals are addressed using a combination of relaxation techniques.  Convex relaxations are used to avoid the combinatorial issues introduced by integer constraints.    A second geometric relaxation technique,  the \textit{workload relaxation} framework of \cite{mey03a,CTCN} is used as an approach to model reduction.  This idea was originally  inspired by the heavy-traffic theory of \cite{harwil87a,law90a,kellaw93a}.

In the research summarized here, the workload relaxation is used for two purposes.  First, it is used to obtain a lower bound $\haeta^*$ on the optimal average cost
for the matching model.  Second,   a value function for the relaxation is used to construct a  real-valued function $h$ on the state space of buffer values.  It is interpreted as an approximate value function for the matching model,  and is used in this paper to define 
a matching policy --- a variant of the
$h$-MaxWeight policy of  \cite{mey09a}.

\Thm{t:hMWao} summarizes the main results of this paper.  A family of arrival processes 
$\{\bfmA^\delta : \delta \in [0,1]\}$ is considered, in which the lower bound   $\haeta^*=\haeta^*(\delta)$ tends to infinity as $\delta\downarrow 0$.   
The performance of the proposed policy
%\snew{name needed!  $h$-MaxWeight policy }
 is shown to be asymptotically optimal, with bounded regret:  
\begin{equation}
\haeta^* \le \eta \le \haeta^*+ O(1),
\label{e:bddregret}
\end{equation}
where the term $O(1)$ is independent of $\delta\in (0,1]$,
and $\haeta^*$ grows as $1/\delta $.

%The bounds obtained in this paper are based on a one-dimensional  worked relaxation.
% \ana{As in Laws' thesis \cite{law90a},  a one-dimensional workload relaxation is defined
%on the same probability space as the original model, and evolves as the controlled random walk, }

The  workload relaxation is a one-dimensional controlled random walk.
For the matching model with arrival process $\bfmA^\delta$, the relaxation
is defined on the same probability space, and evolves as,
\begin{equation}
\haW(t+1) = \haW(t) - \delta + \haI(t) + \Delta(t+1),\qquad t\ge 0,
\label{e:ExWrelaxation}
\end{equation}
in which $\haW(0) \in\Re$ is given,  the idleness process $\haI(t)$ takes values in the interval $[0,\umax]$, 
$\delta>0$ is given,  and $\bfDelta $ is an i.i.d.\ sequence in $\Re$ with zero mean.  The \textit{effective cost} $\barc\colon\Re\to\Re_+$ that is obtained as the value of a nonlinear program.  If the cost function $c$ is linear, then the effective cost is piecewise linear. 
This is taken as a cost function for the one-dimensional workload model.    The lower bound $\haeta^*$
is precisely the optimal average cost for this relaxation,  and admits a tight approximation. 
Details can be found in \Section{s:work}.

Many of the results in \cite{CTCN} on workload relaxations are based on \textit{stabilizability of the arrival-free model}.  That is, it is assumed that the network without arrivals can be stabilized using some policy.  This assumption \textit{fails} for matching models.   Consider the case of organ donation (e.g.  \cite{KidneyPaired}):   if there is a patient waiting for a kidney, and no donors arrive, then the patient will wait  for eternity.  Nevertheless,  there is a natural formulation of workload for these models.  Each component of the multi-dimensional workload process can take on positive and negative values, much like what is found in inventory models.
% \cite[Ch. 7]{CTCN}.   Based on results from this chapter,  
It is found that optimal policies will have structure similar to what is found in inventory theory, such as the classical work of Clark and Scarf \cite{clasca60}.     In particular, based on a one-dimensional relaxation, an approximating 
model is obtained that can be identified as an inventory model of a special form, so that an optimal policy for the relaxation is obtained via a one-dimensional threshold policy.    

These conclusions imply that  optimal policies do not follow the conventions of \cite{busgupmai13}.   Optimal policies may idle, in the sense that no matches are made at certain time instances, even though matches are possible.  

The prior work  \cite{mey09a}  establishes asymptotic optimality of the
$h$-MaxWeight policy
for a class of scheduling models.   In this case the relaxation is a workload model that is non-idling since it evolves on the non-negative integers.    The approximation was logarithmic:  $\eta \le \haeta^*+ O(\log(1/\delta))$.  

%\spmold{April 2:   I commented out the numerical results discussion here,  and added new discussion} 
%Numerical results demonstrate that the $h$-MWT~policy introduced in this paper has much lower cost when compared to polices considered in prior work.  A few examples are given in \Section{s:aOpt}.    *.*******

In the present paper, this policy is refined to take into account structure found in the optimal policy for a workload relaxation, leading to the $h$-MaxWeight with threshold policy, or $h$-MWT.  In the present paper it is shown that the function $h$ can be designed so that the regret is bounded,
in the sense of \eqref{e:bddregret}.
This is the first paper to obtain bounded regret for a stochastic network model, except those special cases for which an optimal policy is known in closed-form.  It is also the first to obtain any form of heavy-traffic approximate optimality when the workload model is not ``minimal'' as in  \cite{mey09a} or  \cite{harwil87a,kellaw93a}.   
 
The prior work \cite{gurwar14} considers a general class of matching models,   with performance analysis based on an asymptotic heavy-traffic setting.  The conclusions are very different because of the different assumptions imposed on the network model:  when specialized to connected bipartite graphs,  their Assumption~1 implies that bipartite graph reduces to a star network.    

The  matching graphs considered in \cite{gurwar14} allow more general topologies, including certain hypergraphs.   Assumption~1 is useful in their analysis because it is then possible to establish a form of path-wise optimality for the workload model under a particular policy.    This is not possible for the models considered in the present paper because the workload process takes on positive and negative values
 --- consequently,  an average-cost optimal policy is a threshold form which is inconsistent with path-wise optimality 
(see Section~5.5 of \cite{CTCN} for further discussion).

%
%The question (ii) has a complete answer in the case that there is a single dominant workload vector in heavy traffic.  For discounted cost, the heavy traffic limit will be a reflected Brownian motion on an interval of the form $[-\threshW,\infty)$,  where $\threshW $ is derived in \cite{wei92a};  \cite{CTCN} obtains the analogous formula for the average cost optimal control problem.  Average-cost optimal control is the focus of this paper.
%
%Heavy traffic theory in cases where the limiting model is not one-dimensional is the object of study in
%\cite{mey05b,mey03a}, as well as \cite{CTCN};   applications to power systems are contained in \cite{chechomey06a,chomey10,meynegwankowsha10}.   Application of these concepts are investigated through numerical experiments.  As predicted in this prior work,   optimal polices can be approximated by variants of threshold policies. 
%

The remainder of the paper is organized as follows:
\Section{s:match} describes the Markovian matching model, the fluid model,
along with a characterization of workload,  and consequences for control. 
This section concludes with the main results for the model in heavy-traffic.  
\Section{s:policy} contains the detailed policy description and the main ideas of the proof. 
Detailed proofs may be found in the Appendix.
Conclusions and directions for future research are described in
\Section{s:conc}.

\section{Bipartite matching model}
\label{s:match}

The  bipartite matching model introduced in this section is 
a
queueing network model with two classes of buffers,  
distinguished by their role as providing supply or demand of  resources.

The description of the model requires the following primitives, where the notation is adapted
from Definition~2.1 of  \cite{busgupmai13}.  We let 
 $\ell_\sply$ denote the number of  supply classes, 
$\ell_\dmd$ denote  the number of demand classes, and define the following index sets:
\smallskip

\begin{tabular}{ll}
$\IDmd$:  Indices of demand classes. \quad  
		& $\ISply$:  Indices of supply classes. 
		\\[.15cm]
 $\IEdge$:     Matching pairs,  $\IEdge\subset \IDmd\times\ISply$. \qquad
 		&
$\IArrive$:    Arrival pairs, $\IArrive\subset \IDmd\times\ISply$
\end{tabular}  

\smallskip

The bipartite graph $(\IDmd\cup\ISply,\IEdge)$  is called the \textit{matching
graph}.  It is assumed throughout that this graph is connected.

The  \NN-network shown in \Fig{f:NN} is an example in which $\ell_\dmd=\ell_\sply=3$,  and the set $\IEdge$ denotes the edges $(e_i)$ shown in the figure.  Each of the three integers $\{x^\dmd_i:i=1,2,3\}$ correspond to units of demand of a particular type, and $\{x^\sply_i: i=1,2,3\}$ correspond to units of the three different types   of supply.

To capture volatility in arrivals and temporal dynamics we introduce next a discrete-time 
Markov Decision Process (MDP) model that resembles a model for a multi-class queueing network.
The main departure from traditional queueing networks is that there are no constraints on service rates.  Instead of ``service'', activities in this model correspond to matching a particular unit of supply with a unit of demand.

\subsection{MDP model}

The vector of buffer levels for the dynamic matching model is denoted $Q(t) $.
It takes values in $\nat^{\ell}$,
where $\ell= \ell_\dmd+\ell_\sply$.
When it is necessary to emphasize the different roles for supply or demand buffers, we use the notation
\begin{equation}
Q(t) =(Q_1^\dmd(t)\,,\dots\,, Q_{\ell_\dmd}^\dmd(t)\,, Q_1^\sply(t)\,,\dots\,,   Q_{\ell_\sply}^\sply(t))^\transpose
\label{e:Qnotation}
\end{equation}
It is often convenient to drop the super-scripts.  In this case, for $i\in \Dmd\eqdef\{1,\dots, \ell_\dmd\}$,  the integer $Q_i(t)$ denotes the number of units of demand of
class $i$,  and for $j\in \Sply\eqdef\{ \ell_\dmd+1,\dots,  \ell_\dmd+ \ell_\sply\}$, the integer $Q_j(t)$  denotes the units of supply of class $j$.

Let  $\xi^0=(1,\dots,1,-1,\dots,-1)^\transpose$,  the vector with $\ell_\dmd$ entries of $+1$, followed by $\ell_\sply$ entries of $-1$.   The queue length vector is subject to the following balance constraint:  
\begin{equation}
\xi^0\cdot Q(t) = 0
\label{e:Qstate}
\end{equation}
For simplicity, in this paper we do not impose upper bounds on buffers.

An i.i.d.\ arrival process is denoted $\bfmA$.  
We adopt the assumptions used in the prior work \cite{busgupmai13},  that a single pair arrive at each time slot -- one of demand and one of supply.   That is, for each $t$,  $A(t)$ takes values in the set  
\begin{equation}
\text{$\dA = \{ \One^{i}+\One^{j} :    ( i,j)\in\IArrive  \}$,}
\label{e:suppA}
\end{equation}
where $\One^{i}$ denotes a column vector with $i$th component equal to $1$ and zero elsewhere. 

An input process $\bfmU$ represents the sequence of matching activities.  The queue dynamics are defined by the recursion,
\begin{equation}
Q(t+1) = Q(t) - U(t) + A(t)\,,\qquad t\ge 0
\label{e:QMM}
\end{equation}   
At each time $t$,  the input is subject to integer constraints, and constraints consistent with the matching
graph.  These constraints are captured by the finite input space, 
\begin{equation}
\dU= \Bigl\{ u= \sum_{e\in \clE} n_e u^e :  \  n_e\in\nat ,\ \  |n|\le \umax\Bigr\}
\label{e:suppU}
\end{equation} 
where $|n| = \sum n_e$, and $\umax\ge 1$ is a fixed integer. 
The vectors $\{ u^e\}$ are an enumeration of all single matches across edges of the matching graph:
that is,  $u^e =   \One^{i}+\One^{j}$ for  $e=( i, j) \in\IEdge$.  
% It issued that $|U(t)|\le \bar u$ for some $\bar u\ge 1$,  where  
There are also implicit  constraints on $U(t) $, since  the components of $Q(t)$ are constrained to non-negative integer values.   

The integer $\umax$ must be chosen sufficiently large to ensure stabilizability of the network.  
This constraint on the input is  imposed only to simplify Taylor series approximations used to obtain performance bounds.   It is assumed henceforth that  $\umax\ge  4$, so that the randomized policies used in our analysis are feasible.

%\spmold{forced to comment an important point here: Consequently, the constraint \eqref{e:Qstate} holds %automatically under \eqref{e:suppU} and \eqref{e:suppA}, provided  it holds ...}
Based on \eqref{e:suppA} and \eqref{e:suppU} we have 
\begin{equation} 
\xi^0\cdot U(t) = 0
\quad\text{\it and} \quad
\xi^0\cdot A(t) = 0\,, \qquad a.s.
\label{e:suppA0}
\end{equation}
Consequently, the constraint \eqref{e:Qstate} holds automatically under \eqref{e:suppU} and \eqref{e:suppA}, provided  it holds at time $t=0$.  
That is, for each $t$,  $Q(t)$ takes values in the set  
$$
\dQ=\{ q\in\nat^\ell : \xi^0\cdot q = 0\}.
$$

The sequence $\bfmU$ is viewed as the \textit{input process} for the MDP model.   Since it is useful to allow $U(t)$ to depend on both $Q(t)$ and $A(t)$, in parts of the analysis the pair process $Y(t) = (Q(t), A(t))$ is chosen as the state process. It is assumed that the input process $\bfmU$ is non-anticipative (a function of present and past values of $\bfmY$).  A stationary (state feedback) policy is of the form $U(t)=\psi(Y(t))$, for some function $\psi\colon\dQ \times \dA \to \dU$. We allow for randomized stationary policies in our analysis, such as \Lemma{t:ZeroDrift} below.

We assume a linear cost function on buffer levels and aim to minimize the average cost given by \eqref{e:eta}, in which $X(t) = Q(t)+A(t)$.  This process evolves very much like \eqref{e:QMM}:
\begin{equation} 
 X(t+1)=X(t) - U(t) + A(t+1), 
\label{e:XMM}
\end{equation}  
which takes values in the discrete space $\dstate=\dQ$.  A stationary policy is denoted $U(t)=\phi(X(t))$ in this case.  The policy in our main result will be of this form.

\Prop{t:YXopt} below establishes that the optimal average cost using the state process $\bfmY$ is identical to what is achievable using the smaller state process $\bfmX$.
The existance of an optimal policy requires stabilizability of the network.

\subsection{Stabilizability}
\label{s:stab}

Let  $\Sply(i) $ denote the set of supply classes that can be matched with a class $i$ demand,  and let $\Dmd(j)$ denote  the set of demand classes that can be matched with a class $j$ supply.  This definition and the extension to subsets $D\subset \IDmd$ and $S\subset \ISply$  is formalized as follows: 
\[
\begin{aligned}
  \Sply(i)&=\{j\in\ISply  : (i,j)\in\IEdge\}
\,,\quad
 \Dmd(j)=\{i\in\IDmd:(i,j)\in\IEdge\}
\\[.2cm]
 \Sply(D)&=\bigcup_{i\in D} \Sply(i)\,,
 		\qquad\qquad  \Dmd(S)=\bigcup_{j\in S} \Dmd(j)
\end{aligned} 
\]
For any vector $x\in\Re^\ell_+$ denote, 
$|x_D| = \sum_{i\in D} x_i 
$,
$
|x_S| =\sum_{j\in S} x_j
$.

%\begin{equation}
%\phi(x) = \argmax\{ u\cdot \nabla h\, (x)  : u\in\dU(x)\},\qquad x\in\dstate.
%\label{e:hMW}
%\end{equation}
%The  MaxWeight policy considered in \cite{busgupmai13} is of this form, with $h(x) =\|x\|^2$, the usual %$\ell_2$-norm.  Under the conditions of this prior work, the MaxWeight policy is stabilizing in the sense %that the controlled MDP model is positive Harris recurrent.   In fact, the policy is stabilizing whenever a %stabilizing policy exists.  

The MDP model is said stabilizable if there exists a policy for which the controlled MDP model is positive Harris recurrent. 

The necessary and sufficient condition for stabilizability of the MDP model is given as follows,
based on the mean arrival rate vector $\alpha=\Expect[A(t)]$: 
\medskip
  {\tt NCond}:   For all non-empty   subsets 
$D\subsetneq \IDmd$ and $S\subsetneq \ISply$,
\begin{equation} 
|\alpha_D | < | \alpha_{\Sply(D)} |  
\quad  \text{\it and}
\quad
| \alpha_S| <  |\alpha_{\Dmd(S)} |   
\label{e:NCond}
\end{equation}
The proof can be found in \cite{busgupmai13}. In that paper, it is shown that the following policy, called Match the Longest (ML) is stabilizing under {\tt NCond}:
\begin{equation}
\phi(x) = \argmax\{ u\cdot \nabla h\, (x)  : u\in\dU(x)\},\qquad x\in\dstate,
\label{e:hMW}
\end{equation}
with $h(x) =\|x\|^2$, the usual $\ell_2$-norm.
Adan and Weiss \cite{adawei12} have shown that the FCFS (First Come First Served) policy also has a maximal stability region. The stationary distribution under FCFS policy has a product form \cite{abbmw13,abmw15}, but there is no efficient algorithm for the normalizing constant. 

The proof of the stability result for ML policy is based on the fact that this policy minimizes the drift of the quadratic Lyapunov function, 
$$
V(q,a) = \sum_{k=1}^{\ell} x_k^2,  
$$
for $(q,a) \in \dQ \times \dA$ and $x = q+a$. A bounded negative drift condition was shown for a randomized policy constructed using network flow arguments. Randomized policies based on network flows are also used in \Section{s:flow}. 
%We give next the full description of this randomized policy and the links with network flow problem as this policy will be our starting point in constructing a new randomized policy that will be used in the proof of our main result, Theorem \ref{t:hMWao} in Section \ref{s:aOpt}.

Under  {\tt NCond} we are also assured of the existence of an optimal policy for either MDP model.
Recall that $\eta^*$ denotes the optimal average for the MDP with state process $\bfmX$.   Let $\eta^*_Y$ denote the optimal average cost for the MDP with the larger state process $\bfmY$.   Given the greater information, it is immediate that $\eta^*_Y\le \eta^*$.  In fact, the two are identical:
\begin{proposition}
\label{t:YXopt}
If   {\tt NCond} holds, then  the optimal average cost exists as a deterministic constant, independent of initial conditions, for either MDP.  
Moreover, the average costs are equal:  $\eta^*_Y = \eta^*$.
\end{proposition}
\begin{proof} 
Theorem~9.0.2 of \cite{CTCN} implies that the existence of optimal policies.   For the MDP with state process $\bfmY$,  it also follows from  \cite[Theorem~9.0.2]{CTCN} that  there is a  function $g^*\colon \dQ \times \dA \to \Re$ that solves the average cost optimality equation:
\begin{equation}
\min_{u\in\dU} \Expect[c(q+a) + g^*(Y(t+1)) \mid Y(t) = (q,a) \, ,\ U(t)=u] = g^*(q,a)+ \eta^*_Y
\label{e:ACOEY}
\end{equation} 
An optimal policy achieving the optimal average cost $\eta^*_Y$ is any minimizer:
\[
\psi^*(q,a) = \argmin_{u\in\dU} \Expect[  g^*(Y(t+1)) \mid Y(t) = (q,a) \, , \  U(t)=u]
\]
Substituting  $ (q-u+a, A(t+1)) =Y(t+1) $ within the conditional expectation gives
$\psi^*(q,a) = \argmin_{u\in\dU} \barg^*(q+a-u) $ where 
\[
\barg^*(q+a-u) = \sum_{a'\in \IArrive}g^*(q+a -u ,a') \Prob\{A(t+1) =a'\}  
\] 

That is,  the optimal input-output pair satisfies $U^*(t) = \psi^*(Y^*(t)) = \argmin_u \barg^*(X^*(t)-u)$,  with $X^*(t)= Q^*(t)+A(t)$ for $t\ge 0$.  This implies that $\eta^*\le \eta^*_Y$, which completes the proof since we already have the reverse inequality.
\end{proof}

\subsection{Workload}
\label{s:work}

For any set $D\subsetneq \IDmd$ we let $\xi^D$ denote the vector whose components are $1$ for $i\in D$,  $-1$ for $i\in\ISply(D)$, and zero elsewhere. The vectors $\{\xi^D\}$ play a role similar to workload vectors in standard queueing models.  Condition~{\tt NCond} can be equivalently expressed,
\[
\xi^D\cdot \alpha < 0,\qquad \text{for all $D\subsetneq \IDmd$ }
\]
We could introduce symmetric notation for $S\subset \ISply$, but this is unnecessary: for each $S\subsetneq \ISply$, 
denote by $\hat{S}=\cup \{S' \; : \; \Dmd(S')=\Dmd(S)\}$.
By definition of $\hat{S}$ and the connectivity of the graph, $\Sply( \Dmd(S)^c ) = \hat{S}^c$. 
Thus, 
\[
\xi^S \leq \xi^{\hat{S}} = \xi^{\Dmd(S)^c} - \xi^0.
\]
%with $\xi^{\emptyset} := (0, \ldots, 0)^T$. 
% there is a set
%$D\subsetneq \IDmd$ such that $\xi^S = \xi^D - \xi^0$.  
Our assumptions imply that $\alpha\cdot \xi^0=0$, so it is sufficient to consider only demand in a characterization of {\tt NCond}.

We can now define a workload process that evolves as \eqref{e:ExWrelaxation}.
For a particular set $D\subsetneq \IDmd$ we take $W(t) = \xi^D\cdot X(t)$,
and $\delta = - \xi^D\cdot \alpha$.  
\begin{proposition}
\label{t:workload}
The workload process evolves according to the recursion,
\begin{equation}
 W(t+1) =  W(t) - \delta + I(t) + \Delta(t+1),\qquad t\ge 0,
\label{e:workload} 
\end{equation}
in which  $\delta>0$,  $\Delta(t+1) = \delta+ \xi^D\cdot A(t+1)$, and $I(t) = - \xi^D\cdot U(t)$.
The zero-mean i.i.d.\ sequence
$\bfDelta $  takes values in $\Re$.   Moreover, $I(t)  $ 
takes values in the non-negative integers $\nat$, and is zero if and only if there is no cross-matching between $\Sply(D)$ and $D^c$. 
\end{proposition}

\begin{proof}
Under  {\tt NCond} it follows that $\delta>0$.  
The properties of $I(t) = -\xi^D\cdot U(t)$ follow from the definition of $\xi^D$. 
\end{proof}

Given a convex cost function $c\colon\Re_+^\ell\to\Re_+$, the \textit{effective cost} 
 is defined as the solution to the convex program,
\begin{equation}
\barc(w)\eqdef\min \{c(x)  : x\in\Re_+^\ell,\  \xi^D\cdot x = w\},\qquad w\in\Re
\label{e:effcost}
\end{equation}
%It is assumed throughout this paper that the cost
% $c$ is linear (see discussion preceding \eqref{e:ACOE}).
It is assumed throughout this paper that $c\colon\Re^\ell_+\to\Re_+$ is a linear function of the state, $c(x) = \sum c_i x_i$, with $c_i>0$ for each $i$.   It follows that $\barc$ is piecewise linear. This and other conclusions are summarized in \Lemma{t:effCost}, whose simple proof is omitted.

\begin{lemma}
\label{t:effCost}
For a linear cost function $c$, the solution to the linear program \eqref{e:effcost}
 results in a piecewise linear function of $w$:
\begin{equation}
\barc(w)=\max(\barc_+ w,  -\barc_- w)
\label{e:barcpl}
\end{equation}
where $\barc_+=\barcpD + \barcpS$ and $\barc_-=\barcmD + \barcmS$,  with 
\begin{equation} 
\begin{aligned}
\barcpD&  = \min\{ c_i : i\in D\}\,,\quad  \barcpS&  = \min\{ c_j : j\in S^c\}  
\\[.15cm]
\barcmS&  = \min\{ c_j : j\in S\}\,,\quad  \barcmD&  = \min\{ c_i : i\in D^c\}   
\end{aligned}
\label{e:barcpl-barc}
\end{equation}
An optimizer $x^*$ for \eqref{e:effcost} exists in which exactly two entries are non-zero.  The form depends on the sign of $w$:  
\begin{equation}
\begin{aligned}
w\ge 0 : \quad 
   x^*_i &=w \quad \text{for some $i\in D$ satisfying $c_i=\barcpD$,} 
    \\
   x^*_j& =w \quad \text{for some $j\in S^c$ satisfying $c_j=\barcpS$.}
 \\[.2cm]
w<0 : 
\quad  
x^*_i &=|w|  \quad \text{for some $i\in D^c$ satisfying $c_i=\barcmD$,} 
\\  
x^*_j&=|w|  \quad \text{for some $j\in S$ satisfying $c_j=\barcmS$.}
\end{aligned}
\label{e:effState}
\end{equation}
\qed
\end{lemma}

The controlled random walk~\eqref{e:ExWrelaxation} with cost function $\barc$ is thus a relaxation of the original MDP model, with controlled input $\bfmhaI$ taking values in $\Re_+$.  
%This model is considered in   \cite[Sec.~7.4 and Sec.~9.7]{CTCN}:
 In \cite[Theorem~9.7.2]{CTCN}  it is shown that an optimal policy  is determined by a threshold policy of the following form:  There is a scalar $\threshW>0$ so that 
\begin{equation}
\haI(t) = \max\{\delta - \haW(t)  - \threshW ,0 \}
 \label{e:OptIdle}
\end{equation}
Under this policy, the stochastic process $\{\Phi(t) = \haW(t)-\Delta(t)\}$  is a reflected random walk on   $[-\threshW,\infty)$.   Equation~(7.37) of \cite{CTCN} defines the diffusion heuristic, intended to approximate this threshold based on a reflected-Brownian motion (RBM) model,   
 \begin{equation}
\thresh  =  \half \frac{\sigma^2_\Delta }{\delta}  \log\Bigl(1+\frac{\barc_+}{\barc_-}\Bigr)
\label{e:barrdiffusion}
\end{equation}
where $\delta $ is the drift appearing in  \eqref{e:OptIdle} and  $\sigma^2_\Delta$  is the variance of $ \Delta(t)$.
% (also the variance of $\xi^D\cdot A(t)$).
%,  and the cost parameters are given in \eqref{e:barcpl}.  

%The RBM model is the focus of heavy traffic theory e.g.~\cite{harwil87a,kellaw93a}.  

%\subsection{Connections with network flows}
\subsection{Randomized policies and network flows}
\label{s:flow}

In \cite[Theorem 7.1]{busgupmai13}, a randomized policy based on a solution of a network flow problem was used in the stability proof for the %MaxWeight 
ML policy. The randomized policy used in the proof of our main result, Theorem \ref{t:hMWao}, will be based on a slighlty different network flow problem, corresponding to a specific workload vector. 

We use 
 standard terminology 
 from
  network flow theory. 
Fix some demand and supply sets $D\subset \IDmd$ and  $S \subset \ISply$,
and let
$\resG$ denote the restricted matching graph:   
$\resG = \{D\cup S , \resEdge\}$, where $\resEdge = \{(i,j) \in\IEdge \; : \; i \in D \textrm{ and } j \in S\}.$
Consider the directed graph
\begin{equation}\label{eq-dg}
{\clN}_{D \cup S} = \bigl(D\cup  S  \cup \{a,f\}, 
{\clE}_{\clN_{D \cup S}}
%\resEdge\cup\{(a,i), i\in D\} \cup \{(j,f), j\in\ISply\} 
\bigr)\:  ,
\end{equation}
where ${\clE}_{\clN_{D \cup S}}= \resEdge\cup\{(a,i), i\in
D\} \cup \{(j,f), j\in S \}$. 
Nodes $a$ and $f$ are the source and the sink of this network, as illustrated in Figure \ref{fig:maxflow}.
Endow the arcs of $\resEdge$ with infinite capacity, an arc of type $(a,i)$
with capacity $\alpha_i$, and an arc of type $(j,f)$ with capacity
$\alpha_j$. 
 
\begin{figure}[htb]
   \centering
   \includegraphics[width=.8\textwidth]{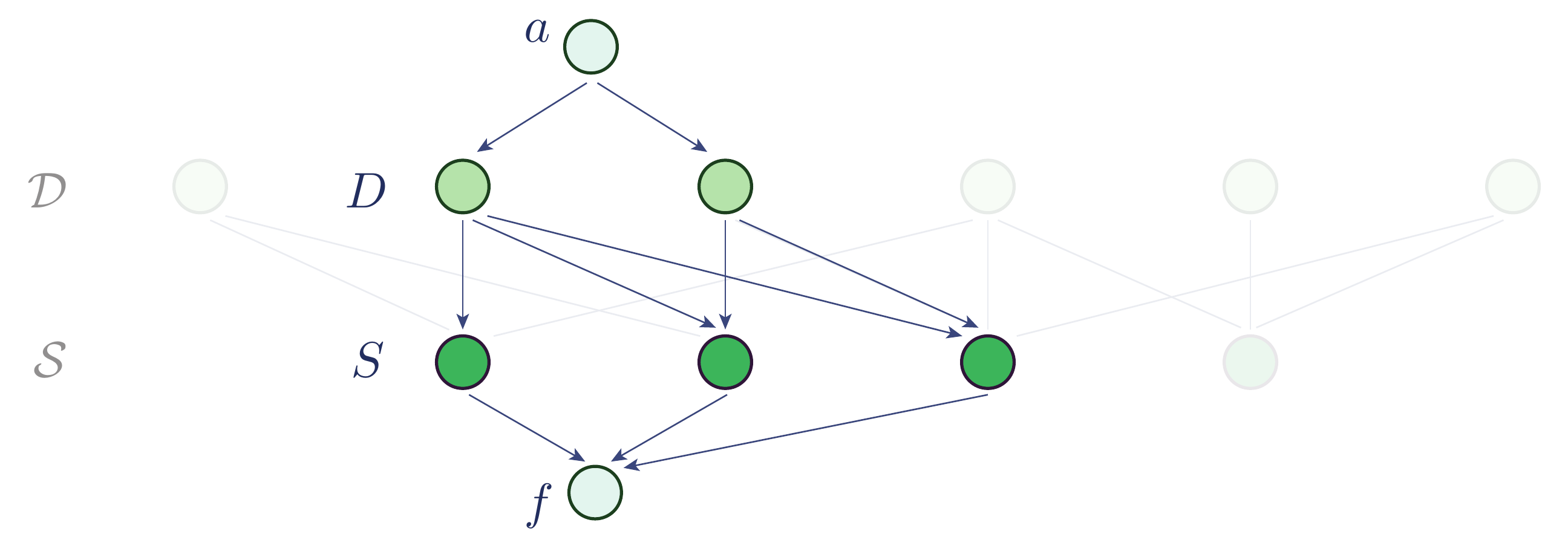}
   \caption{\label{fig:maxflow}Network flow problem for the directed graph
   ${\clE}_{\clN_{D \cup S}} $: The arcs in $ \resEdge$
   %$\IEdge$  
have infinite capacity, an arc of type $(a,i)$ has capacity $\alpha_i$, and an arc of type $(j,f)$ has capacity $\alpha_j$}
 \end{figure}

Recall that a {\em  cut} is a subset of the arcs
whose removal disconnects $a$ and $f$. The {\em capacity} of a cut is
the sum of the capacities of the arcs.

Recall that $T: \mathcal{E_N} \rightarrow \Re_+$ is a
{\em flow} if: (i) $\forall
i\in D $,  $T(a,i)=\sum_{j\in S (i)} T(i,j)$,  and $ \forall j \in S $, $\sum_{i\in
  \Dmd(j)}T(i,j)=T(j,f)$; 
(ii) $\forall (x,y) \in \mathcal{E_N}$, $T(x,y)$ is less
  than
   or
equal to the capacity of $(x,y)$. 
The {\em value} of $T$ is $\sum_{i}
T(a,i)=\sum_{j}T(j,f)$. 

\begin{lemma}{\cite{busgupmai13}}
\label{t:MCMF}
If $\Sply(D) \subset S$, then under conditions (\ref{e:NCond}), the maximal $a$-$f$ flow is equal to $\alpha_D$.  Moreover, there exists a maximal $a$-$f$  flow such that the flow value on each edge in $\resEdge$ is strictly positive.
Symetrically, if $\Dmd(S) \subset D$, then there is a strictly positive $f$-$a$ flow of value $\alpha_S$. 
\end{lemma}

\begin{proof}
The result for $D = \IDmd$ and $S = \ISply$ is stated as Lemma 3.2 in \cite{busgupmai13}. 
%\cite[Lemma 3.2]{busgupmai13}
% and this more general result can be found in the proof of Theorem 7.1. in \cite{busgupmai13}. 
Denote by $F$ a strictly positive flow of value $\alpha_{\IDmd} = \alpha_{\ISply} = 1$ for network ${\clN}_{\IDmd \cup \ISply}$ and by $F_{D \cup S}$ its restriction to the arcs in ${\clN}_{D \cup S}$. 
Flow $F_{D \cup S}$ is strictly positive and has value $\alpha_D$ (all the outgoing arcs from set $D$ in the original network are still present in ${\clN}_{D \cup S}$, as $\ISply(D)\subset S$).   
\end{proof}

%\medskip

The following construction will be one of the key ingredients in the proof of our main result, Theorem \ref{t:hMWao}. It is assumed that conditions \eqref{e:NCond} hold. 
\paragraph{Translating a network flow into a randomized policy}

Let $D\subset \IDmd$ and $\delta = - \xi^D\cdot \alpha$. 
Let $F$ be a strictly positive flow of value $\alpha_D$ for the network flow problem for the matching graph $\resG$, with $S = \Sply(D)$.
%$\Sply(D) \subset S$.   

A randomized matching policy can be  defined as follows. Assume that the state at time $t$ is $Y(t)=(Q(t), A(t)) = (q, a)$ for some $(q,a) \in \dQ \times \dA$.   
Denote by $D^+(q) = \{ i \in D \mid q_i > 0\}$ and $S^+(q) = \{ j \in S \mid q_j > 0\}$ respectively the set of demand and the supply queues with non-empty buffers.  
For each $j \in S$, we define a probability vector $p^{(j)}$, % of size $|D^+(q)|$, 
\begin{equation}
\label{e:pi}
p_{i}^{(j)}= \frac{1}{\alpha_j} \Bigl[ \ F(i,j) + \frac{\alpha_j-F(j,f)}{|D^+(q)|} \ \Bigr],\quad i \in D^+(q)\:.
\end{equation}
This defines indeed a probability vector:
\begin{eqnarray*}
\sum_{i\in D^+(q)} p_{i}^{(j)} & = &  \frac{1}{\alpha_j} \Bigl[ \ 
\alpha_j-F(j,f) + \sum_{i\in D^+(q)} F(i,j) \ \Bigr] \\
& = &  \frac{1}{\alpha_j} \Bigl[ \ \alpha_j-F(j,f)  + F(j,f) \ \Bigr]
\ \  =
  \ \  1 \:.
\end{eqnarray*}
Symmetricaly, for each $i \in D$, we define a probability vector $p^{(i)}$, %of size $|S^+(q)|$, 
\begin{equation}
\label{e:pj}
p_{j}^{(i)}= \frac{1}{\alpha_i} \Bigl[ \ F(i,j) + \frac{\alpha_i-F(a,i)}{|S^+(q)|} \ \Bigr],\quad j \in S^+(q)\:.
\end{equation}
As $F$ is a network flow, these probability vectors satisfy, 
\begin{alignat}{2}
p_{i}^{(j)}& \geq \frac{F(i,j)}{\alpha_j}, \qquad &j\in S,  \   i\in D^+(q), 
\label{eq:pbounds}
\\
p_{j}^{(i)} &\geq \frac{F(i,j)}{\alpha_i}, \; &i\in D, \  j\in S^+(q).
\end{alignat}

Let $a = \One^{i}+\One^{j}$ for some $(i,j)\in\IArrive$.  Then the new arrival of a supply item $j$ is matched to a demand item in $D^+(q)$ independently of
the past, according to probability vector $p^{(j)}$. 
The new arrival of a demand item $i$ is matched to a supply item in $S^+(q)$ independently of
the past, according to probability vector $p^{(i)}$. 

\begin{lemma}
\label{t:ZeroDrift}
Under this randomized policy, the following drift condition holds whenever $i\in D  $ such that $q_i \ge 1$:   
\begin{equation} 
\Expect\bigl[ Q_i (t+1) \mid Q(t)=q\bigr] 
\le  q_i.
\label{eq:qi}
\end{equation}
%{\color{red} The same null drift holds for supply buffers $j\in S  $ satisfying $q_j\ge 1$.}
For supply buffers $j\in S  $ satisfying $q_j\ge 1$,
\begin{equation}
\Expect\bigl[ Q_j (t+1) \mid Q(t)=q\bigr] 
\le  q_j + \alpha_j - \sum_{i \in D \cap \Dmd(j)} F(i,j) \le q_j + \delta
\label{eq:qj}
\end{equation}
and 
\begin{equation}
\label{eq:sqj}
\Expect\bigl[ \sum_{q_j \geq 1} Q_j (t+1) \mid Q(t)=q\bigr] 
\le  \sum_{q_j \geq 1}  q_j + \delta.
\end{equation}
\end{lemma}

\begin{proof}
Consider $i\in D^+(q)$.  Under this policy we obtain a lower bound on the conditional mean $ \Expect [ U_i(t) \mid Q(t)=q, A(t) = a] $: It is minimized when $q_{\ell} \ge 1$ for every demand buffer $\ell \in D $,  and $q_k = 0$ for every supply buffer $k$ such that $(i,k)\in \resEdge$.   
%\sean{avoid use of $\ell$;
%I used this for the number of nodes in the graph}
There is at most one supply arrival at each time step, $\sum_{j\in S } a_j \leq 1$ (the supply arrival can be in class $S^c$). 

%If $a_j = 1$ (a supply arrival of class $j$), then the edge $(i,j)$ is chosen with probability that is 
%lower bounded by 
% the conditional probability that $(i,j)$ is chosen, knowing that an edge in $\{(\ell,j)\; : \; \ell \in D \}$ is chosen. 

Consequently,   
\[
 \Expect [ U_i(t) \mid Q(t)=q, A(t) = a] =
\sum_{j \in S (i)} \ind\{ a_j =1 \} p_{i}^{(j)}\geq\sum_{j \in S (i)} \ind\{ a_j =1 \} \frac{F(i,j)}{\alpha_j}.
\] 
Recalling that $\Expect\bigl[ A_i(t) \mid Q(t)=q\bigr] = \Expect\bigl[ A_i(t) \bigr] = \alpha_i$, 
we conclude that 
\begin{eqnarray*}
\Expect\bigl[  U_i(t) \mid Q(t)=q\bigr] & = & \Expect\bigl[ \Expect [ U_i(t) \mid Q(t), A(t) ] \mid Q(t)=q\bigr] \\
& \geq & \sum_{j \in S (i)} \alpha_j \frac{F(i,j)}{\alpha_j}  =  \sum_{j \in S (i)}F(i,j) = \alpha_i,
\end{eqnarray*}
where the final equality follows because  $F$ is a maximum value flow. 
The claim \eqref{eq:qi} follows since,
\[
\begin{aligned}
\Expect\bigl[ Q_i (t+1) \mid Q(t)=q\bigr]
&= 
\Expect\bigl[ Q_i (t) + A_i(t) - U_i(t) \mid Q(t)=q\bigr] 
\\
&= q_i+ \alpha_i -
\Expect\bigl[   U_i(t) \mid Q(t)=q\bigr] 
\end{aligned}
\]

The proof of (\ref{eq:qj})  is similar.  The last inequality in (\ref{eq:qj}) and (\ref{eq:sqj}) follow from  the fact that for all $j$, 
$\sum_{i \in D \cap \Dmd(j)} F(i,j) \leq \alpha_j$, and 
$$
\sum_{j \in S} \Bigl (\alpha_j - \sum_{i \in D \cap \Dmd(j)} F(i,j) \Bigr) = \alpha_S - \alpha_D = \delta.
$$
\end{proof}

\subsection{Asymptotic optimality}
\label{s:aOpt}

The structure of the policy for the relaxation is the inspiration for the following refinement of the  $h$-MaxWeight policy in \eqref{e:hMW}.

For a   differentiable function $h\colon\Re^\ell\to\Re_+$, and a threshold $\tau\ge 0$,
the $h$-MWT ($h$-MaxWeight with threshold) policy is obtained as the solution to the constrained non-linear program,
 \begin{equation}
\begin{aligned}
\phi(x)  = \argmax\quad & u\cdot  \nabla h\, (x)  
\\[.15cm]
\text{subject to}\quad & u\in\dU(x)  \quad \text{\it and} \quad  \xi^D\cdot  (x-u) \le -\tau
\end{aligned}
\label{e:wchMW}
\end{equation}
Based on the definition of workload and  idleness appearing \eqref{e:workload},  the constraint $\xi^D\cdot  (x-u) \le -\tau$ is equivalently expressed
\[
 I(t) \le -W(t) - \tau,\qquad \text{\it when  $X(t)=x$ and $U(t)=u$.}
\]
This constraint recalls the definition of a threshold policy  \eqref{e:OptIdle}
 for the workload relaxation.   
 
We take $\tau=\thresh$ in our main results,  and the function $h$ is also designed using inspiration from the workload relaxation.

\spm{March 2016.  For future reference:  Remember we want to prove this claim
\\
The integer program \eqref{e:wchMW} appears to be difficult to solve for large networks.  Fortunately it admits a convex relaxation that is tight:}

To evaluate performance we consider an asymptotic setting:
Assume that we have a family of arrival processes $\{ A^\delta(t) \}$  parameterized by 
$\delta\in [0,\maxdelta]$, where $\maxdelta\in (0,1)$.  Each is assumed to satisfy \eqref{e:suppA}.
The following additional assumptions are imposed throughout:
\begin{romannum}
\item[(A1)]
For one set $D\subsetneq \IDmd$ we have $\xi^D\cdot \alpha^\delta =-\delta$,   where $\alpha^\delta$ denotes the mean of $A^\delta(t) $.

Moreover, there is a fixed constant $\udelta>0$ such that  $\xi^{D'}\cdot \alpha^\delta \le -\udelta$ for any $D'\subsetneq \IDmd$, $D'\neq D$, and $\delta\in [0,\maxdelta]$.

\item[(A2)]  
The distributions are continuous at $\delta=0$, with linear rate:  For some constant $b$,
\begin{equation}
\Expect[ \| A^\delta(t) - A^0(t)\|] \le b\delta.
\label{e:CtsVar}
\end{equation}

\item[(A3)]  
The sets $\IEdge$  and $\IArrive$ do not depend upon $\delta$,  and the graph associated with $\IEdge$ is connected. 
Moreover, there exists $i_0\in \Sply(D)$, $j_0\in D^c$, and $\pcm >0$ such that 
\begin{equation}
 \Prob\{A^\delta_{i_0}(t) \ge 1 \ \text{\it and} \ A^\delta_{j_0}(t) \ge 1\} \ge \pcm ,\qquad  0\le \delta\le \maxdelta.
\label{e:IdleA}
\end{equation} 

\end{romannum}
We   suppress the dependency of $\bfmA,\bfmQ,\bfmU$ on $\delta$ when there is no risk of confusion.
We also let $\xi=\xi^D$, so that $\delta=-\xi \cdot\alpha$.   
%$\delta=-\xi^\transpose\alpha$.   

We are now prepared to state the main result of the paper, establishing asymptotic optimality of a family of $h$-MWT policies.   
The construction of the function $h$ is performed in \Section{s:policy}.  We let $\eta^*$ denote the optimal average cost for the MDP model,  
$\haeta^*$ the optimal average cost for \eqref{e:ExWrelaxation},  and the following is shown to approximate each of these values:
\begin{equation}
\haeta^{**} =  \thresh \barc_-   = \half \frac{ \sigma^2_\Delta }{ \delta} \, \barc_-  \log\Bigl(1+ \frac{ \barc_+}{\barc_-} \Bigr) 
\label{e:etastar}
\end{equation}

\begin{theorem}[Asymptotic Optimality With Bounded Regret]
\label{t:hMWao}
Suppose that Assumptions (A1)--(A3) hold.   For each $\delta\in (0,\maxdelta]$, there is a function $h$ such that the $h$-MWT policy using the threshold $\thresh$ has finite average cost  $\eta$,  satisfying the following bounds, 
\[
\haeta^*\le \eta^*
\le 
\eta\le \haeta^* +O(1)
\]
where the constant $O(1)$ does not depend upon $\delta$.  Moreover,  the average cost for the relaxation is approximated by the value in \eqref{e:etastar}:
\[
 \haeta^* = \haeta^{**} +O(1)
\]
%\qed
\end{theorem} 
%\begin{theorem}[Asymptotic Optimality With Bounded Regret]
%\label{t:hMWao}
%Under   Assumptions (A1)--(A3),  for  
%sufficiently large $\kappa>0$,  $\beta>0$, and  sufficiently small $\delta_+>0$ (each independent of $\delta$),  the average cost  $\eta$ under the \wchMW~policy satisfies,
%\snewer{can these bounds go in the abstract?}
%\[
%\haeta^*\le \eta^*
%\le 
%\eta\le \haeta^* +O(1)
%\]
%where $\eta^*$ is the optimal average cost for the MDP model,  $\haeta^*$ is the optimal average cost for \eqref{e:ExWrelaxation},
%and the constant $O(1)$ does not depend upon $\delta$.  Moreover,  the average cost for the relaxation satisfies the uniform bound,
%\[
% \haeta^* = \haeta^{**} +O(1)
%\]
%where average cost on the right hand side is defined in \eqref{e:etastar}. 
%%\qed
%\end{theorem} 

The proof is constructive, with $h$ defined in eq.~\eqref{e:hConstructed} in the following section.
 
\section{Construction of the h-MWT policy}
\label{s:policy}
 
In what follows we describe a particular construction of $h$ designed to approximate the solution to an average cost optimality equation (ACOE) for the MDP model:  
For $x\in\dstate$,
\begin{equation}
\min_{u\in\dU} \Expect[c(X(t)) + h^*(X(t+1)) \mid X(t)=x\, ,\ U(t)=u] = h^*(x) + \eta^*, %\qquad x\in\dstate,
\label{e:ACOE}
\end{equation}
in which $\eta^*$ is the optimal average cost, and $h^*$ is the relative value function.  
The construction is based on a heavy-traffic setting, following the work of Harrison \cite{harwil87a}, Kelly \cite{kellaw93a} and subsequent research (see \cite{CTCN} for a bibliography).  It is taken as the sum of two terms,  
\[
h(x) = \hah(\xi\cdot x) + h_c(x)
\]
Similar to \cite{mey09a}, the function $h_c$ is introduced to  penalize deviations between $c(x)$ and $\barc(\xi\cdot x)$. 

The first term $\hah$ is a function of workload.  
For $w\ge -\thresh$, it   solves the second-order differential equation,
\begin{equation}
-\delta   \hah'\, (w) +\half\sigma^2_\Delta   \hah''\, (w)  = - \barc(w) +\haeta^{**} , 
\label{e:RBMACOE}
\end{equation} 
where the threshold $\thresh$ is defined in  \eqref{e:barrdiffusion},
and
the optimal average cost $\haeta^{**} $ is given in \eqref{e:etastar}. 
There  is a solution that is convex and increasing on $[-\thresh,\infty)$,  with $\hah'(-\thresh) = \hah''(-\thresh) = 0$:
For constants $\{A_\pm,B_\pm,C_\pm,D_\pm\}$,
\begin{equation}
\hah(w) =
\begin{cases}  
A_+ w^2 + B_+ w   &  w\ge0
\\[.25cm]
A_- w^2 + B_- w + C_- + D_- e^{\Theta w}  &  -\thresh\le w\le 0
\end{cases}
\label{e:hah}
\end{equation}  
where $   {\Theta}^{-1} =\half   \sigma^2_\Delta/ \delta $.  Formulae for all of the other parameters can be found in the proof of \Lemma{t:RBMACOE}. 

The domain is extended to obtain a convex, $C^2$ function on all of $\Re$.  
We fix a parameter $\delta_+ \in (0,\pcm )$, where $\pcm >0$ is used in
Assumption~(A3).  The sum $\delta+\delta_+$ is interpreted as the desired idleness rate when $ w< -\thresh $. 
Fix another constant $\theta>0$, and for $ w < -\thresh $ define
\begin{equation}
\hah(w) = 
\hah(-\thresh)   
+
\frac{\barc_-}{\delta_+} \Bigl[ \half (w+\thresh)^2 + \frac{1}{\theta} (w+\thresh)  + \frac{1}{\theta^{2}} \Bigl( 1-\exp(\theta (w+\thresh) ) \Bigr)\Bigr]
\label{e:hahneg}
\end{equation}
where   $\hah(-\thresh)    $ is given in \eqref{e:hah}.  

We now turn to the construction of $h_c$.     For this we might take a constant times $ [c(x) - \barc(\xi\cdot x)]^2 $.   This fails because of positive drift on the boundary of $\dstate$.

Let $\tilx$ denote the function of $x$ with entries,
\spm{March 2016.  would be nice if $i$ meant demand, $j$ supply,  and $k$ anything!}
\[
\tilx_k = x_k + \beta( e^{-x_k/\beta} - 1) \, ,
\]
where $\beta>0$ is a constant.   The right hand side vanishes at the origin, as does its first derivative.  The constant $\beta$ is chosen so that its derivative with respect to $x_k$ is  small for $x_k$ in some interval $[0,\barq]$ (see Step 2 of the proof of \Theorem{t:hMWao} below).

A similar transformation for workload is used,
\begin{equation}
\tilw = \sign(w) \bigl[ |w| + \beta( e^{-|w|/\beta} - 1)  \bigr]
\label{e:tilw}
\end{equation}
If $ w=\xi\cdot  x$ then the definition does not change, but $\tilw$ is of course a function of $x$; the perturbation ensures that $\barc(\tilw)$ is $C^1$ as a function of $x$.
 
The resulting perturbation of $ [c(x) - \barc(w)]^2 $ is used to define $h_c$.  Consequently,  $h$ and its gradient are given by, 
\begin{alignat}{2}
h(x)  = \hah(w) + h_c(x) &=  \hah(w) + \kappa [c( \tilx) - \barc(\tilw)]^2 &
\label{e:hConstructed}
\\
\nabla h\,(x) &= \hah'(w) \xi + \nabla h_c(x)\,, \qquad  &x\in\Re_+^\ell, \  w=\xi\cdot x.
\label{e:hparts}
\end{alignat}

\begin{proof}[Proof of  \Theorem{t:hMWao}]

In what follows we give the  proof of \Theorem{t:hMWao} using the function $h$ in \eqref{e:hConstructed},  with sufficiently large values of $\kappa$ and $\beta$,
and a sufficiently small value of $\delta_+>0$.   
The notation $O(1)$ implies a term that is uniformly bounded in both $\delta>0$ and $q\in\dstate$.  This bound may depend on $\beta$, $\kappa$, and $\delta_+$,
but these constants will be fixed, independent of $\delta$ and $q$.

\wham{Step 1}
The first step is to obtain a bound that suggests the ACOE \eqref{e:ACOE}.
However, for performance bounds it is more convenient to work with the Markov chain $\bfmQ$ rather than $\bfmX$.  The function $h$ is translated to a function of $Q(t)$ via,
\[
V(q) =  \Expect[   h(X(t)) \mid Q(t)=q] =
\Expect[h(q+A(t))],
\] 
where the value of $t$ is arbitrary. 

Recall from \Prop{t:workload} that $I(t)=0$ if and only if there is no cross-matching between $\Sply(D)$ and $D^c$,  where   $I(t) = - \xi^\transpose   U(t)$.   From the constraints in \eqref{e:wchMW} it follows that $I(t) = 0$ when $W(t)\ge -\thresh$ under the $h$-MWT policy.

 \Lemma{t:suggestsACOE} follows from a second-order Taylor series expansion for $h$.  It also requires the uniform Lipschitz bound for $\hah''$ that is given in \Lemma{t:RBMACOE}.

\begin{lemma}
\label{t:suggestsACOE}
The following bound holds under any policy satisfying $I(t) = 0$ when $W(t)\ge -\thresh$:  For each  $  q\in\dstate$,  
\begin{equation}
\begin{aligned}
 \Expect[  V(Q(t+1)) -& V(Q(t)) \mid Q(t)=q]   \le  
 -  \cIDP(q)      + \eta^* +  O(1)\, ,
% \\[.2cm]
% &= \Expect[  h(X(t+1)) - h(X(t)) \mid Q(t)=q]  
% \le  
% -  \cIDP(q)      + \eta^* +  O(1)\, ,
\end{aligned}
\label{e:MWbdd}
\end{equation}
where the function $\cIDP$ is   dependent on the policy:
\begin{equation}
 \cIDP(q)  \eqdef   \Expect[  \nabla h\, (X(t))  \cdot (U(t) -\alpha)   \mid Q(t)=q \bigr]  
 +\half \sigma^2_\Delta \hah''(\xi\cdot q) - \eta^*
 \label{e:cmw} 
\end{equation}
\qed
\end{lemma}

The subscript ``idp'' refers to ``inverse dynamic programming'';  a concept from approximate dynamic programming  in which an approximate value function is given, and from this a cost function is obtained \cite{put14}.   If this cost function approximates the cost function of interest, then via this approach we obtain bounds on the average cost.

Suppose that the following bound can be established for \textit{some policy} satisfying  
the non-idling constraint in \eqref{e:wchMW}:
\begin{equation} 
  \cIDP(q)      \ge c(q)  + O(1) \, .
\label{e:GammaPolicyBdd}
\end{equation}
The $h$-MWT policy maximizes the dot-product within the expectation in \eqref{e:cmw},
subject to  \eqref{e:wchMW}.    Consequently, if \eqref{e:GammaPolicyBdd} holds for one such policy, it must also hold for the $h$-MWT policy.

The bounds \eqref{e:MWbdd}  and \eqref{e:GammaPolicyBdd}
 imply   the \textit{Foster-Lyapunov drift condition}
\begin{equation}
\Expect[  V(Q(t+1)) - V(Q(t)) \mid Q(t)=q] 
 \le  
 - c(q)
 + \eta^* + O(1)
 ,\qquad q\in\dstate.
\label{e:V3}
\end{equation}
The sequence $\bfmQ$ is a Markov chain under any stationary policy.
It is well-known that \eqref{e:V3}  implies that the average-cost defined in \eqref{e:eta}
 is bounded by $\eta^* + O(1)$ \cite{MT,CTCN}.

To see how the bound \eqref{e:GammaPolicyBdd} is obtained, recall the definition of $h$ in \eqref{e:hparts} as the sum of two terms.  Considering each term separately leads to the proof of \Lemma{t:cmwDecomposition}, which bounds $\cIDP$ by the sum of two terms. The proof is given in \Section{s:cmwDecomposition}.

This bound requires two additional definitions: 
mean idleness, and the gap between the current cost and the effective cost,
\begin{equation}
\begin{aligned}
\barIdle(q) &= \Expect[   I(t)    \mid Q(t)=q  ]  
  \\
  \zeta(q) & = c(q) - \barc(\xi\cdot q) 
\end{aligned}
\label{e:barIdle-zeta}
\end{equation}

\begin{lemma}
\label{t:cmwDecomposition}
The following decomposition holds under the assumptions of \Lemma{t:suggestsACOE}, 
\begin{equation}
-\cIDP(q)      \le B(q)+B_c(q) +O(1)
\label{e:t:cmwDecomposition}
\end{equation}
in which  
%\begin{eqnarray}
%B(q) &=&   (\barIdle(q) -\delta) \hah'(\xi\cdot q) +\half \sigma^2_\Delta \hah''(\xi\cdot q) -\haeta^{**} 
%\label{e:cmwA}
%\\[.25cm]
%B_c(q) &=& 2\kappa \zeta(q) \Bigl\{
%\Expect\bigl[    c(\tilX(t+1)) - c(\tilX(t))     \mid Q(t)=q \bigr]      + \barc_-   \barIdle(q) +\barc_+\delta \Bigr\} 
%\label{e:cmwB}
%\end{eqnarray}
%
\begin{align}
B(q) &=   (\barIdle(q) -\delta) \hah'(\xi\cdot q) +\half \sigma^2_\Delta \hah''(\xi\cdot q) -\haeta^{**} 
\label{e:cmwA}
\\[.25cm]
B_c(q) &= 2\kappa \zeta(q) \Big\{
\Expect\bigl[    c(\tilX(t+1)) 
\label{e:cmwB}
\\
&\hspace{3cm}- c(\tilX(t))     \mid Q(t)=q \bigr]      + \barc_-   \barIdle(q) +\barc_+\delta \Bigr\} 
\notag
\end{align}
\qed
\end{lemma}

In Step 2   we construct a randomized policy, designed so that 
\begin{romannum}
\item The non-idling constraint in \eqref{e:wchMW} holds.
\item  Idling is enforced for $W(t)<-\thresh$:
\begin{equation}
\barIdle(q) =(\delta+\delta_+ ) \Expect[      \ind\{W(t)<-\thresh\} \mid Q(t)=q\bigr] \,.  
\label{e:EnforceIdle}
\end{equation} 
It is shown in \Lemma{t:e:EnforceIdle} (in \Section{s:cmwDecomposition} of the appendix) that\eqref{e:EnforceIdle} implies the following bound:
 \begin{equation}
B(q)  \le  -\barc(\xi\cdot q) +  O(1)
\label{e:Abdd}
\end{equation}
The proof is based on the dynamic programming equation \eqref{e:RBMACOE}.

\item For some $\epsy_c>0$, independent of $\delta>0$,
\begin{equation}
B_c(q)  \le - 2\kappa \bigl(\epsy_c - \barc_- \delta_+ - \barc_+\delta \bigr) \bigl(c(q) - \barc(\xi\cdot q)\bigr) 
\label{e:Bbdd}
\end{equation}
\end{romannum}
The constants $\delta_+$ and $\maxdelta$ are chosen so that $\epsy_c - \barc_- \delta_+ -\barc_+\maxdelta>0$.  We then take $\kappa = \half (\epsy_c - \barc_- \delta_+-\barc_+\maxdelta)^{-1}$, so that
\[
B_c(q)  \le -   c(q)   + \barc(\xi\cdot q)   + O(1)
\]
 
Substituting these bounds on $B(q)$ and $B_c(q)$ into \eqref{e:t:cmwDecomposition}, we conclude that the drift condition \eqref{e:GammaPolicyBdd} holds under this policy.  The $h$-MWT~policy must also satisfy the   lower bound \eqref{e:GammaPolicyBdd}, which establishes  \eqref{e:V3}.

 \medskip

\wham{Step 2} 
A randomized policy is designed to satisfy properties (i)--(iii) from Step 1. In particular,   idling is permitted only when the workload is below the threshold $-\thresh$.     

A constant $\barq\ge1$ is required in the construction of this policy, and based on this we define
\begin{equation}
\minc = 1+4\ell\barq  \cmax 
\label{e:minc}
\end{equation}
where $\cmax=\max c_i$, and $\ell$ is the number of buffers in the network.  
The lower bound required for $\beta $ is given in
\eqref{e:beta}, and a lower bound required for $\barq$ is given in 
\eqref{e:barq}.

\begin{lemma}
\label{t:RandomizedHeavyDrift}
Under the assumptions of \Thm{t:hMWao},   
 there exist constants  $\bardelta_0\in (0,\maxdelta)$ and
$\epsy_0, \epsy'_0>0$, 
such for each $\delta \in [0,\bardelta_0]$, there is
a randomized policy that allows no cross-matching, and satisfies the following uniform bounds:  For any $q\in\dstate$,
\begin{romannum}
\item
If  $c(q) < \barc(\xi\cdot q)+ \minc $, then for each $j$
\begin{equation}
\Expect\bigl[   Q_j (t+1) \mid Q(t)=q\bigr] \le  q_j  + \delta \,, \qquad \text{provided $q_j\ge 1$}
\label{e:RandomizedHeavyDriftA}
\end{equation}

\item
If  $c(q)\ge \barc(\xi\cdot q) + \minc$, then
there are $k,m$ for which $q_k\ge \barq$,  $q_m\ge 1$,  
\begin{equation}
\begin{aligned}
\Expect\bigl[     Q_k (t+1)     \mid Q(t)=q\bigr] &\le  q_k  
% + \delta 
-\epsy_0
\\[.25cm]
\Expect\bigl[   c_k Q_k (t+1)  +c_m Q_m (t+1)   \mid Q(t)=q\bigr] &\le  c_k q_k+c_mq_m 
%+ \max\{c_k, c_m\} \delta
-\epsy'_0 
\end{aligned}
\label{e:RandomizedHeavyDriftB}
\end{equation} 
and moreover \eqref{e:RandomizedHeavyDriftA} continues to hold for all $j\neq m$.  \end{romannum}
\qed
\end{lemma}
 
A parallel result is established in the Appendix (\Lemma{t:RandomizedHeavyDriftIdle}), in which idling is enforced  at average rate $\delta+\delta_+$ when $W(t)<-\thresh$, so that \eqref{e:EnforceIdle} holds.

The bound \eqref{e:Bbdd} requires an additional step.   

\wham{Step 3} 
Whenever
$c(q)\ge \barc(\xi\cdot q) + \minc$,  the randomized policy satisfies,
\begin{equation}
\Expect\bigl[ c(\tilX(t+1)) -c(\tilX(t))   \mid Q(t)\bigr]
\le 
-\epsy_c     
\label{e:tilWdrift}
\end{equation}
where $\epsy_c>0$.
The proof is contained in \Section{s:e:tilWdrift};  it is based on the proof of
Prop.~2.7 of \cite{mey09a}.        

These three steps define the randomized policy that satisfies \eqref{e:GammaPolicyBdd},
which shows that \eqref{e:V3}  holds under the $h$-MWT policy.   
This completes the proof of the theorem. 
\end{proof}

\begin{figure}
\Ebox{1}{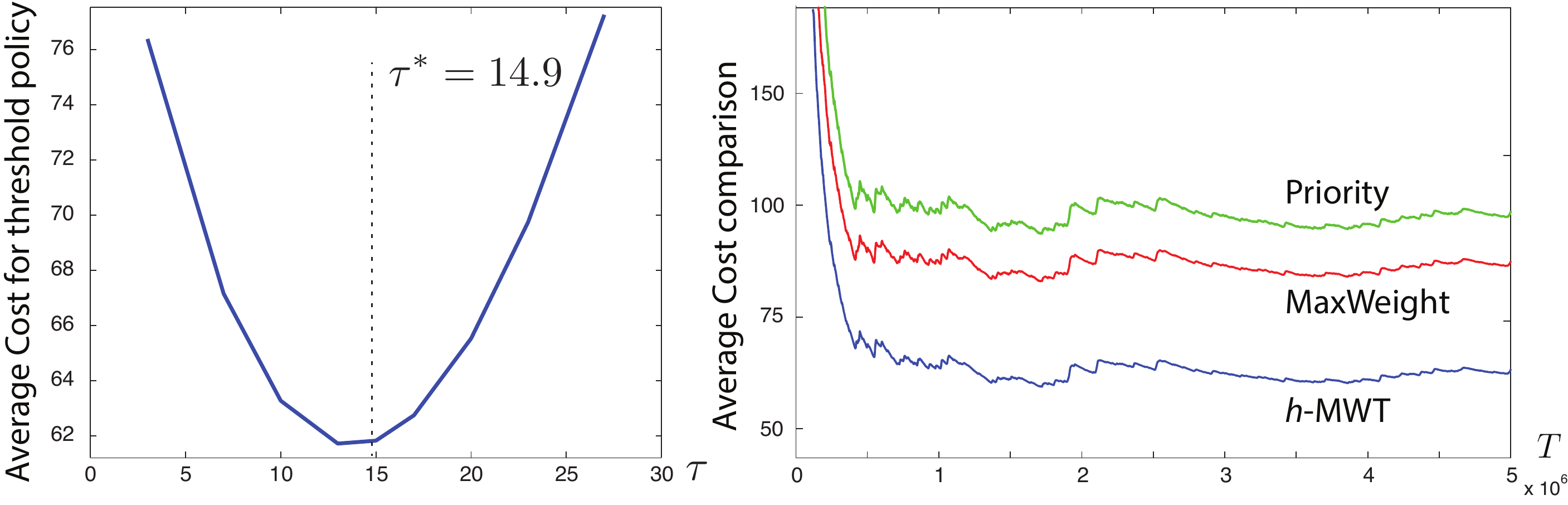}
\vspace{-.35cm}
\caption{\small
Shown on the right is the $h$-MWT~policy in which the threshold $\tau$ was taken as a parameter.     The plots on the right hand side compare the average cost obtained using $h$-MWT,  MaxWeight (given by \eqref{e:hMW} with $h(x) = \sum c_i x_i^2$),  
and a priority policy.}  
%\snewest{Sean adjusted figure, only to remove threshold.  What do we call this!?  \qquad NOTE:  We should give some flexibility in the definition of the policy.  We don't have to take $\tau$ equal to the RBM value.  I guess it could be called \color{blue} $h/\tau$-MW policy.}
\label{f:123321}
\end{figure}

%The  proof reveals that the $h$-MaxWeight policy is a greedy policy that attempts to empty ``expensive'' buffers myopically, subject to the constraint that cross-matching between $S$ and $D^c$ is permitted only when $W(t)\le -\thresh$.  
%\snew{changes needed here?  I don't know.  Let's discuss}  

Numerical experiments were performed for the \NN-network for various cost functions and arrival statistics.  Results from one set of experiments are shown in \Fig{f:123321}.  The set $D$ in this experiment was taken to be $D =\{3^{\dmd}\}$,
and the arrival rate was chosen so that $\delta = -\xi^\transpose\alpha$ was much smaller than $ -\xi^{D'}\cdot\alpha$ for any other set $D'\subsetneq\Dmd$.
The cost   was taken to be $c(x)  =  x_1^\dmd + 2x_2^\dmd+3 x_3^\dmd +3
x_1^\sply + 2x_2^\sply +x_3^\sply$.  

The policy \eqref{e:wchMW} was considered for various values of the threshold $\tau$.   
For each value of $\tau$, the average cost was estimated by the sample-path averages,
\[
\frac{1}{T} \sum_{t=1}^T c(Q(t))\, .
\] 
The value $ T= 5\times 10^6$ was required for reliable estimation.

 The results are shown on the left hand side of 
\Fig{f:123321}. 
The average cost is large because the drift vector was taken to be small,
  $\delta = 0.007$.    Recall that the workload relaxation \eqref{e:ExWrelaxation}
 will have an average cost of order $O(\delta^{-1})$.  The best value of $\tau$ obtained through simulation is very close to the value $\thresh$ predicted by the RBM model.

In the comparison plots shown on the right, the
 static priority policy gives priority to vertical matches (edges $e_1$, $e_3$ and $e_5$ in \Fig{f:NN}).    The MaxWeight policy considered was cost-weighted:
 Given the state $x$, a new demand of type $i$ is matched to a supply class $j^*$ satisfying,
\[ 
j^*\in \argmax_{j\in \ISply(i)} c_jx_j^\sply.
\] 
Matching of  supply is determined symmetrically.

The average cost under the $h$-MWT~policy (using the   $\tau=\thresh$) performed the best -- about 30\%\ lower than the cost-weighted MaxWeight policy.

\section{Conclusions}
\label{s:conc}

The dynamic bipartite matching model is an unusual queueing system with the particularity that the workload process can be negative. 
%potentially complex system. 
We have shown how relaxation techniques can lead to insight for the construction of good policies with low complexity.  

The numerical results show that the average-cost performance can be outstanding when compared with priority policies, or standard versions of the MaxWeight policy. 
It is remarkable how well the ``diffusion heuristic'' predicts the best threshold for the discrete-time model.

The key argument is a correspondence with models in inventory theory. Although the theoretical results are based on a heavy-traffic setting, this structure  will play some role even when the assumptions of the paper are violated.

%In current research we are considering workload relaxations of dimensions greater than one where the %threshold policy is replaced by switching curves in workload space.

%imsart-nameyear.bst,imsart-number.bst
 
\bigskip

\noindent
{\bf Acknowledgments:} This research is supported by the NSF grants CPS-0931416 and CPS-1259040,
the French National Research Agency grant ANR-12-MONU-0019.

\bibliographystyle{abbrv} 
%\bibliography{strings,markov,q,bibs/matching}
\bibliography{bibs/strings,bibs/markov,bibs/q,bibs/sim,bibs/matching}

%\clearpage

\vspace{1cm}

\appendix

\centerline{\Large\bf Appendix}
\bigskip

This appendix contains full details for the first two steps of the proof of \Thm{t:hMWao}.  The third step is immediate from Prop.~2.7 of \cite{mey09a}.

\section{Step 1:  Drift for $h$}

This section concerns mainly the function $\hah$,  and concludes with a proof of \Lemma{t:cmwDecomposition}. 

Recall that $\hah$ depends on the parameter $\Theta = \delta/(2 \sigma^2_\Delta )$, and a fixed constant $\theta>0$.  
Based on the solution we obtain these conclusions:
\begin{proposition}
\label{t:RBMACOEa}
There exist unique constants $\{A_\pm,B_\pm, C_-,D_-\}$ such that the function $\hah$ defined by (\ref{e:hah},\ref{e:hahneg}) 
satisfies the following:
\begin{romannum}
\item
The ODE \eqref{e:RBMACOE} holds for $w>-\thresh$,  and for $w\le -\thresh$ we have the approximation,
\begin{equation}
\delta_+
\hah'\, (w) +\half\sigma^2_\Delta   \hah''\, (w) 
= -\barc(w) + \haeta^{**}  
%\frac{\barc_-}{\theta} + 
%\half \sigma^2_\Delta  \frac{\barc_-}{\delta_+}  
+  O(1)
\label{e:RBMACOE-}
\end{equation}
where the term ``$O(1)$'' depends on $\theta$ but is independent of $\delta$.

\item $\hah$ is strictly convex and $C^2$ on all of $\Re$,  with a unique minimum at the threshold.  Moreover,  
\[
\hah'(-\thresh) = \hah''(-\thresh) = 0.
\]  

\item
It has quadratic growth globally,  and cubic growth locally:  For a constant  $K<\infty$, and all $\delta>0$,
\[ 
|\hah'(w) | \le  K \min \Bigl\{   \delta^{-1} (w+\thresh),    (w+\thresh)^2 \Bigr\}
\]

\end{romannum}
\end{proposition}

%Note that the term $\half \sigma^2_\Delta  \frac{\barc_-}{\delta_+}$ in \eqref{e:RBMACOE-}  is also $O(1)$.
%It is included in the approximation only to show that there may be an additional cost when $\delta_+ $ is small.

\begin{proof}
\Lemma{t:RBMACOE} establishes that $\hah$ is $C^2$ and solves the ODE \eqref{e:RBMACOE}  on the domain
$(-\thresh,\infty)$.  For $w\le -\thresh$ we use the definition \eqref{e:hahneg} to compute,
\begin{equation}
\begin{aligned} 
\delta_+
  \hah'\, (w) +&\half\sigma^2_\Delta \hah''\, (w) 
\\
 &= \delta_+ \frac{\barc_-}{\delta_+} \Bigl[ (w+\thresh) + \frac{1}{\theta}   - \frac{1}{\theta}  \exp(\theta (w+\thresh)  )\Bigr]  
 \\
 &\qquad + \half \sigma^2_\Delta  \frac{\barc_-}{\delta_+} \Bigl[ 1   -    \exp(\theta (w+\thresh)  \Bigr]   
\end{aligned}
\label{e:RBMACOEneg}
\end{equation}
On using the identity $\barc_- \thresh = \haeta^{**}$, and $\barc_-w = -\barc(w)$ for $w<0$, the right hand side becomes,
\[
\begin{aligned}
 -\barc(w) + \haeta^{**}  +  \bigl( \frac{\barc_-}{\theta}    + \half \sigma^2_\Delta  \frac{\barc_-}{\delta_+} \bigr) \bigl(  1   -    \exp(\theta (w+\thresh)   \bigr) 
\end{aligned}
\]
This   proves  \eqref{e:RBMACOE-} when $  w+\thresh \le 0$. 

We now prove (ii). It   follows from the definition \eqref{e:hahneg} that $\hah$ is convex on $(-\infty,-\thresh)$ --- its second derivative is strictly positive on this interval.  It also follows from this definition that 
\[
\lim_{w\uparrow -\thresh}\hah'(w) = \lim_{w\uparrow -\thresh}\hah''(w) = 0.
\]
 It is established in 
\Lemma{t:RBMACOE}  that $\hah$ is $C^2$,
 convex and increasing on $(-\thresh,\infty)$, with
\[
\lim_{w\downarrow -\thresh}\hah'(w) = \lim_{w\downarrow -\thresh}\hah''(w) = 0.
\]
Part (iii) also follows from \Lemma{t:RBMACOE}.
%This is established in \Lemma{t:RBMACOE}.
\end{proof}

\subsection{Computation of $\hah$ for $w\ge -\thresh$}

The following result gives properties of $\hah$ on this domain.   The proof contains an explicit construction for $\hah$.

\begin{lemma}
\label{t:RBMACOE}
There exist unique constants $\{A_\pm,B_\pm, C_-,D_-\}$ for which the function \eqref{e:hah}  is $C^2$ on $[-\thresh,\infty)$,
with $ \hah'\, (-\thresh) =\hah''\, (-\thresh) =0$.   With these parameters, 
the function $\hah$ has the following additional properties:
\begin{romannum}
\item It is  strictly convex on $[-\thresh,\infty)$.

\item    The second derivative satisfies, for some $K<\infty$, and all $\delta>0$,  $w,w'\in[-\thresh,\infty)$,
\[
\begin{aligned}
     \Bigl| \frac{d^2}{dw^2}\hah(w)     \Bigr| & \le K \frac{1}{ \delta}
\\[.25cm]
   \Bigl| \frac{d^2}{dw^2}\hah(w)  -\frac{d^2}{dw^2}\hah(w')  \Bigr| &\le K |w-w'|
\end{aligned}
\] 
\item The third derivative exists for $w\neq 0$, it is uniformly bounded by $(\barc_++\barc_-)/\sigma^2_\Delta$ on $(-\thresh,0)$,  and vanishes for $w>0$.
\end{romannum}
\end{lemma}

\proof
We first demonstrate that for a unique choice of parameters, 
the function $\hah$ is
a
$C^1$ solution to \eqref{e:RBMACOE}.  It immediately follows that the function is also 
$C^2$ on $[-\thresh,\infty)$, since  the ODE \eqref{e:RBMACOE} implies the following representation for the second derivative:  
\begin{equation}
\half\sigma^2_\Delta \hah''\, (w)   = 
\delta \hah'\, (w)  - \barc(w) + \barc(-\thresh),\qquad w\ge -\thresh
\label{e:C1C2}
\end{equation}
The last term appears because $\haeta^{**} =  \barc(-\thresh)$ by the definitions.

For $w>0$,  eq.~\eqref{e:RBMACOE} gives,
\[
-\delta (2A_+ w  +B_+ )  +\half\sigma^2_\Delta  (2A_+)   = -\barc_+ w + \haeta^{**}
\]
from which we conclude that 
\[
A_+ = \frac{\barc_+}{2\delta},\qquad B_+ = \frac{1}{\delta}\Bigl( \sigma^2_\Delta A_+ -\haeta^{**}\Bigr)   
\]
In terms of $\Theta$ this becomes,
\[
A_+ =\frac{1}{\Theta}\frac{1}{\sigma^2_\Delta} \barc_+ ,\qquad
  B_+ =  \frac{2}{\Theta} A_+  - \frac{1}{\delta}\haeta^{**}
\]
For $w<0$ there is the additional exponential term,  and the right hand side is modified as follows,
\[
-\delta (2A_- w  +B_- + D_- \Theta e^{\Theta w}  )  +\half\sigma^2_\Delta  (2A_- + D_- \Theta^2 e^{\Theta w})  = \barc_- w + \haeta^{**}
\]
The exponential terms on the left hand side cancel, which gives as previously,
\[
A_- = - \frac{1}{\Theta}\frac{1}{\sigma^2_\Delta} \barc_-
,\qquad B_- =   \frac{2}{\Theta} A_-  - \frac{1}{\delta}\haeta^{**}
\]

The parameter $D_-$ is computed by imposing the constraint that $\hah$ is differentiable at the origin:
\[
B_- + \Theta D_-
=
\frac{d}{dw} \hah\, (0-) =\frac{d}{dw} \hah\, (0+) = B_+
\]
Consequently,
\[
D_- = \frac{1}{\Theta} \Bigl( B_+ - B_- \Bigr)= \frac{2}{\Theta^2} ( A_+ + | A_-|)  =  \frac{2}{\Theta^3}  \frac{1}{\sigma^2_\Delta}  (\barc_++\barc_-)
\]
We then obtain $C_-$ by imposing continuity at zero.
%By construction $\hah$ is $C^1$ on this domain, and thence by \eqref{e:C1C2}  it is also $C^2$.
 
The second derivative is given by,
\[
\frac{d^2}{dw^2}\hah(w) =  2A_- +\Theta^2 D_- e^{\Theta w},\qquad -\thresh<w<0\, .
\]
The right hand side evaluated at $-\thresh$ becomes,
\[
\frac{d^2}{dw^2}\hah\,(-\thresh) =   2A_- +\Theta^2 D_- e^{-\Theta \thresh} =0
\]
This follows from the formula $\thresh=\Theta^{-1}\log(1+\barc_+/\barc_-)$,  and the formulae for $A_-$ and $D_-$.  
It is thus established that $ \hah''\, (-\thresh) = \hah'\, (-\thresh) =0$.

Next we establish convexity.  For this, consider the  third derivative,
\[
\frac{d^3}{dw^3}\hah(w) = \Theta^3 D_- e^{\Theta w}\,,\qquad -\thresh<w<0.
\]
Hence the second derivative is increasing on this domain, and we have seen that $ \frac{d^2}{dw^2}\hah\,(-\thresh) =  0$.   
It follows that the second derivative is strictly positive on $(-\thresh,0)$.  
The second derivative is obviously positive on $\Re_+$,  which implies strict convexity on $[-\thresh,\infty)$.

Finally,  the third derivative is bounded by $\Theta^3 D_- = (\barc_++\barc_-)/\sigma^2_\Delta$, which establishes (iii),
and the Lipschitz property for the second derivative in (ii).
\qed

\subsection{Implications for  workload}
\label{s:a-work}

The following bounds are obtained using a second order Taylor-series approximation.
combined with the uniform Lipschitz continuity of $\hah''$ obtained in  \Lemma{t:RBMACOE}.

The first general bound is expressed in terms of the ``ideal'' idleness process,
\[
\text{
 $I^0(t)= 0$ if $W(t)\ge -\thresh$, and $I^0(t)= \delta_+ +\delta$ otherwise.}
 \]
Recall that $\delta_+$ is used in the definition of $\hah$ on the interval $(-\infty,-\thresh)$,  and $\pcm $ was introduced in (A3).
\begin{proposition}
\label{t:CBMCRWneg}
Consider the workload process in discrete time defined by $W(t) = \xi^\transpose X(t)$, which evolves as \eqref{e:workload}.     
\begin{romannum}
\item 
The following bounds hold under any policy satisfying $I(t) = 0$ 
when $W(t)\ge -\thresh$:  For each  $  q\in\dstate$,
\begin{equation}
\begin{aligned}
 \Expect[\hah(W(t+1))  - & \hah(W(t))  \mid Q(t)=q]   
 \\
 &=  -\barc(\xi\cdot q)  + \haeta^{**}   
 \\
 \qquad
 &
 +  \Expect[ \hah'(W(t)) ( I(t) - I^0(t) )  \mid Q(t)=q \bigr]  + O(1)
\end{aligned}
\label{e:TShah+}
\end{equation} 
%where $I(t)= -\xi^\transpose U(t)$, and the term ``$O(1)$'' is independent of the input and of $\delta>0$.

\item
Assume that  the scalar $\delta_+>0$ used in the definition of $\hah$ satisfies $\delta_+\in(0, \pcm )$.
There is a stationary policy   such that for each $q\in\dstate$ and $\delta\in (0, \pcm -\delta_+)$,
\begin{equation}
  \Expect[\hah(W(t+1))  - \hah(W(t))  \mid Q(t)=q] \le  -\barc(\xi\cdot q) + \haeta^{**}  + O(1)
\label{e:ACOEbarcneg}
\end{equation}

\end{romannum}
\end{proposition}

\begin{proof}
\Prop{t:RBMACOEa} establishes that $\hah$ is $C^2$ and convex, with a unique minimum at $-\thresh$.   
The second derivative $\hah''$ satisfies a Lipschitz bound that is independent of $\delta$,  by \Lemma{t:RBMACOE}.  
These results make possible a second-order Taylor series approximation
to bound the drift \eqref{e:TShah+}: 
\[
\begin{aligned}
 \Expect[\hah(&W(t+1))  -  \hah(W(t))  \mid Q(t)=q ]   
 \\
 &=  \Expect\bigl[     \hah'(W(t)) [-\delta + I(t) + \Delta(t+1)]  \mid Q(t)=q ]   
 \\
 &\qquad\qquad 
 +  \Expect\bigl[ 
   \half \hah''(W(t)) [-\delta + I(t)+\Delta(t+1)]^2  \mid Q(t)=q ]    + O(1)
\end{aligned}
\]

where each expectation on the right hand side is conditioned on   $  Q(t)=q$.
These two terms on the right hand side can be transformed  using the fact that
$\Delta(t+1)$ has zero mean, and is independent of $(Q(t),W(t))$: 
\[
\begin{aligned}
 \Expect\bigl[     \hah'(W(t)) [-\delta  & + I(t) + \Delta(t+1)]   \mid Q(t)=q \bigr] 
 \\
 & =  \Expect\bigl[     \hah'(W(t)) [-\delta + I(t) ]   \mid Q(t)=q \bigr] 
 \\[.2cm]
\Expect\bigl[    \hah''(W(t)) [-\delta & + I(t)+\Delta(t+1)]^2   \mid Q(t)=q \bigr]  
\\
   &  = \Expect\bigl[ 
     \hah''(W(t)) [-\delta + I(t)]^2   \mid Q(t)=q \bigr] 
     \\
     &\qquad  + \sigma^2_\Delta\Expect\bigl[ 
     \hah''(W(t))    \mid Q(t)=q \bigr]  
\end{aligned}
\]
We have $ \Expect\bigl[ \hah''(W(t)) [-\delta + I(t)]^2   \mid Q(t)=q \bigr] =O(1)$ under the conditions of (i), so that on combining these three equations,
\[
\begin{aligned}
 \Expect[\hah(W(t+1))  - & \hah(W(t))  \mid Q(t)=q]   
    \\[.2cm] 
    &=\Expect\bigl[     \hah'(W(t)) [-\delta + I(t)  ]   \mid Q(t)=q \bigr] 
 \\
 &\qquad 
 +  \half \sigma^2_\Delta  \Expect\bigl[ 
    \hah''(W(t))    \mid Q(t)=q \bigr]   + O(1)
\end{aligned}
\]
     
Consider the special case $I(t)=I^0(t)$.  
Combining the ODE  \eqref{e:RBMACOE} and the ODE bound given in \eqref{e:RBMACOE-} gives
\[
\begin{aligned}
 \Expect[\hah( & W(t+1))  -  \hah(W(t))  \mid Q(t)=q]   
         \\[.2cm]
 &=  \Expect\bigl[     \hah'(W(t)) [-\delta + I^0(t)  ]  +  \half \sigma^2_\Delta \hah''(W(t))    \mid Q(t)=q \bigr]  + O(1)
    \\[.2cm]
   &=
    -
 \Expect\bigl[     \barc(W(t))    \mid Q(t)=q \bigr]  +\haeta^{**}   + O(1)   \\[.2cm]
   &=
    -
 \barc(\xi\cdot q)    +\haeta^{**}   + O(1)
\end{aligned}
\]
This is why  $I^0(t)$ is called ``ideal''.

In general we have    additional terms because of the error between $I(t)$ and $I^0(t)$:
\[
\begin{aligned}
 \Expect[\hah(W(t+1))  - & \hah(W(t))  \mid Q(t)=q]   
 \\
 &=  -\barc(\xi\cdot q)  + \haeta^{**}   
 +  \Expect[ \hah'(W(t)) ( I(t) - I^0(t) )  \mid Q(t)=q ]   + O(1)
\end{aligned}
\]
which gives \eqref{e:TShah+}.

\medskip

We now prove (ii).  Consider any policy with the following two features:  First,
$I(t) = 0$ whenever $W(t)\ge -\thresh$.  Second, when $W(t)< -\thresh$, then exactly one cross-match on the edge $(i_0,j_0)$ is performed whenever the two corresponding buffers are non-empty:  $X_{i_0}(t)\ge 1$ and $X_{j_0}(t)\ge 1$.  This event occurs with positive probability due to the bound \eqref{e:IdleA}
assumed in (A3).

Under a policy with these two features,  a first order Taylor series approximation gives 
 the simpler approximation,
\begin{equation}
  \begin{aligned}
 \Expect[\hah( & W(t+1))  -  \hah(W(t))  \mid Q(t)=q]   
 \\
 &=  -\barc(\xi\cdot q)  + \haeta^{**}   
 \\
 &
 \quad
 +\hah'(\xi\cdot q )   \Expect[     ( I(t) -  I^0(t)  )  \ind\{W(t)<-\thresh\}    \mid Q(t)=q ]   + O(1)
\end{aligned}
\label{e:CBMCRWnegProof-ii}
\end{equation}

We consider two cases separately.  

First we consider $q$ for which $w(q) \ge -\thresh$, so that $\hah'(\xi\cdot q )\ge 0$.
If $\xi^\transpose (q+A(t)) = W(t)<-\thresh$,  then $w(q) <-\thresh+1$ since $\xi^\transpose A(t)\ge -1$.    It follows that 
$|w(q)+\thresh |\le 1$ if $W(t)<-\thresh$ occurs with positive probability,  and  $w(q) \ge -\thresh$.  Moreover,   
using a   Taylor series expansion,
\[
\hah'(\xi\cdot q ) =  \hah'(-\thresh) + \hah''(-\thresh) (\xi\cdot q+\thresh) +O(1)
\]
where we have used    the Lipschitz bound on $\hah''$.  The first and second derivatives of $\hah$ vanish at the threshold,
giving $\hah'(\xi\cdot q ) = O(1)$.  This and \eqref{e:CBMCRWnegProof-ii}
 establish the desired bound in (ii) for $q$ satisfying $w(q) \ge -\thresh$.  

Consider next $q$ satisfying $\hah'(\xi\cdot q ) \le 0$.   Equivalently, $w(q) \le -\thresh$.
 In this case we have the lower bound,
 \[
 I(t)   \ind\{W(t)<0\}  \ge  \ind\{ A(t) = (i_0, j_0) \}
 \]
 Hence by
\eqref{e:CBMCRWnegProof-ii},  when $\hah'(\xi\cdot q ) \le 0$,
\[
  \begin{aligned}
 \Expect[\hah( & W(t+1))  -  \hah(W(t))  \mid Q(t)=q]   
 \\
 & \le   -\barc(\xi\cdot q)  + \haeta^{**}   
 +\hah'(\xi\cdot q )   \Expect[     \ind\{ A(t) = (i_0, j_0) \} -\delta_+-\delta ]   + O(1)
 \\
 &\le  -\barc(\xi\cdot q)  + \haeta^{**}   
 + \hah'(\xi\cdot q ) [\pcm -\delta_+-\delta ] + O(1) \\
 &\le  -\barc(\xi\cdot q)  + \haeta^{**}    + O(1)
\end{aligned}
\] 
 \end{proof}

Similar calculations show that $\haeta^{**} $ is an approximate lower bound on $\haeta^*$, and hence also $\eta^*$:
\begin{lemma}
\label{t:heats}
The    average cost $ \haeta^*$ for the workload model is approximately lower bounded by its approximation, uniformly in $\delta>0$:
\[
 \haeta^*  \ge \haeta^{**} +O(1).
\]
\end{lemma}

\begin{proof}
The optimized relaxation is denoted  $\bfmhaW^*$:
 this is the controlled Markov chain \eqref{e:ExWrelaxation},  with optimal idleness process given in \eqref{e:OptIdle}.  The increment $-\delta + \Delta(t) = \xi^\transpose A(t)$ takes values in $\{-1,0,1\}$,  so that $\haW^*(t)\ge -\threshW-1$ for $t\ge 1$.    

Optimality of the threshold policy    \eqref{e:OptIdle}
is established in
\cite[Theorem~9.7.2]{CTCN} through a construction of the solution $\hah^*$ to the ACOE,
\[
\min_{I\ge 0} \Expect[ \hah^*(  w - \delta + I + \Delta(t+1) ) ] = \hah^*(w)  -\barc(w) +\bareta^*,\qquad w\in\Re.  
\]
The minimum is achieved using \eqref{e:OptIdle}.  This can be expressed,
\[
 \Expect[\hah^*(\haW^*(t+1)) -  \hah^*(\haW^*(t))  \mid \haW^*(t)] =  -\barc(\haW^*(t)) +\bareta^*,\quad  t\ge 0. 
\]
The relative value function $\hah^*$ has quadratic growth on $(0,\infty)$, and $\bareta^*$ is the optimal average cost.  
By summing the equation above over $t=0$ to $N-1$,  and taking the  expectation of both sides, we obtain for each initial condition $ \haW^*(0)=w\in\Re$,
\[
\frac{1}{N} \sum_{t=0}^{N-1} \Expect[\barc(\haW^*(t)) ] 
= \bareta^* + \frac{1}{N} \Bigl( \hah^*(w) -  \Expect[\hah^*(\haW^*(N))  ] \Bigr)
\]
The right hand side converges to $
\bareta^*$ as $N\to\infty$.  It follows that $  N^{-1}    \Expect[\hah^*(\haW^*(N)) \mid \haW^*(0)=w]$ also tends to zero.
The function $\hah$ also has quadratic growth, which implies  
\begin{equation}
\lim_{N\to\infty}  N^{-1}    \Expect[\hah(\haW^*(N)) \mid \haW^*(0)=w] = 0,\qquad w\in\Re.  
\label{e:hahNull}
\end{equation}

%  it follows that $ \hah^*$ is linear on a semi-infinite interval:
%\[
%\hah^*(w)  = 
%  \Expect[ \hah^*(  -\threshW + \Delta(t+1) ) ]  -\barc w- \bareta^*,\qquad w\le \delta-\threshW.
%\]
%where we have substituted $\barc(w) = -\barc w$, which is valid for $w<0$.

Following the same steps as in the proof of \Prop{t:CBMCRWneg},   we can show 
that for each $w \in\Re$,  each $\delta>0$,  and any 
idleness process $\{\haI(t)\}$,
\[
\begin{aligned}
 \Expect[\hah( & \haW(t+1))  -  \hah(\haW(t))  \mid \haW(t)=w]   
 \\
 &=  -\barc(w)  + \haeta^{**}   
 +  \Expect[ \hah'(\haW(t)) ( \haI(t) - I^0(t) )  \mid \haW(t)=w \bigr] 
 \\
 &\qquad   
  +
   \half   \Expect  [   \hah''(\haW(t)) [ I(t) ] ^2  \ind\{ W(t)\ge -\thresh\}  \mid \haW(t)=w \bigr]   + O(1)
\end{aligned}
\]
Hence for the optimal process,
\[
\begin{aligned}
 \Expect[\hah( & \haW^*(t+1))  -  \hah(\haW^*(t))  \mid \haW^*(t)=w]   
 \\
 &\ge  -\barc(w)  + \haeta^{**}   
 +  \Expect[ \hah'(\haW(t)) ( \haI^*(t) - I^0(t) )  \mid \haW^*(t)=w \bigr]  + O(1)
\end{aligned}
\]
This bound holds for any choice of $\delta_+ >0$ in the definition of $\hah$ and $\bfmhaI^0$;  
we don't require $\delta_+ \in (0,\pcm )$.

We have $\haI^*(t)\le 1$ for all $t$, for  initial conditions $w\ge-\threshW$.  If we choose $\delta_+=1$ in the definition of $\hah$,  it follows that $ \haI^*(t) - I^0(t) \le 0$  when $\hah'(\haW(t)) <0$,  and hence $\hah'(\haW(t)) ( \haI^*(t) - I^0(t) ) \ge 0$ for all $t$.  That is,  
\[
\begin{aligned}
 \Expect[\hah(\haW^*(t+1))  - & \hah(\haW^*(t))  \mid \haW^*(t)=w]   
 \\
 &\ge  -\barc(w)  + \haeta^{**}     + O(1)\,,\qquad w\ge-\threshW.
\end{aligned}
\]
Summing both sides over $t=0$ to $N-1$ as before, we obtain the lower bound,
\[
\begin{aligned}
\frac{1}{N} \sum_{t=0}^{N-1} \Expect[\barc(\haW^*(t)) ] 
&\ge \bareta^{**} + \frac{1}{N} \Bigl( \hah(w) -  \Expect[\hah(\haW^*(N)) ] \Bigr)  +O(1)
\end{aligned}
\]
in which $  \haW^*(0)=w$ in each expectation.
Finally, applying the result
\eqref{e:hahNull}
gives,
\[
\begin{aligned}
\bareta^*  &= \lim_{N\to\infty} \frac{1}{N} \sum_{t=0}^{N-1} \Expect[\barc(\haW^*(t)) \mid \haW^*(0)=w] 
\\
& \ge \bareta^{**} + \lim_{N\to\infty} \frac{1}{N}   \Bigl( \hah(w) -  \Expect[\hah(\haW^*(N)) \mid \haW^*(0)=w] \Bigr)  +O(1)
 \\
& = \bareta^{**} +O(1)
\end{aligned}
\] 
\end{proof}

\subsection{\Lemma{t:cmwDecomposition} and related bounds}
\label{s:cmwDecomposition}

The results in this subsection begin with the 
representation 
\[
\Expect[  V(Q(t+1)) -  V(Q(t)) \mid Q(t)=q]  
= \Expect[  h(X(t+1)) - h(X(t)) \mid Q(t)=q] , 
\]
combined with 
the
following Taylor series approximation:
\begin{subequations}
\begin{align}
\Expect[  h(X(t+1)) - &h(X(t)) \mid Q(t)=q] 
\nonumber
\\
&=  \Expect[  \hah'(W(t)) [W(t+1) - W(t) ] \mid Q(t)=q]  
\label{e:hahp}\\
&\qquad +  \Expect[  \hah''(W(t)) [W(t+1) - W(t) ]^2 \mid Q(t)=q]  
\label{e:hahpp} 
\\
&\quad\qquad +  \Expect[  \nabla  h_c\, (X(t))   \cdot (U(t) +\alpha)   \mid Q(t) \bigr]  +O(1)
\label{e:nabhc}
\end{align}
\end{subequations}
The $O(1)$ error arises from   Lipschitz continuity of $\hah''$ obtained in  \Lemma{t:RBMACOE}, which justifies evaluating the second derivative at $W(t)$ in \eqref{e:hahpp}.   

In the proof that follows it is shown that the function  $B_c(q)$ is an approximation of \eqref{e:nabhc},
and $B(q) +\haeta^{**}$ is an approximation of the sum of \eqref{e:hahp}  and  \eqref{e:hahpp}:
\begin{equation}
\begin{aligned}
B(q) + \haeta^{**} &=  \Expect[\hah(W(t+1))  -   \hah(W(t))  \mid Q(t)=q]   +O(1)
\\[.2cm]
&=
 \Expect[  \hah'(W(t)) [W(t+1) - W(t) ] \mid Q(t)=q] 
 \\
 &\qquad    +  \Expect[  \hah''(W(t)) [W(t+1) - W(t) ]^2 \mid Q(t)=q]  
+O(1) 
\end{aligned}
\label{e:cmwAapp}
\end{equation}
 
\begin{proof}[Proof of \eqref{e:cmwAapp} and  \Lemma{t:cmwDecomposition}]  
To establish \eqref{e:cmwAapp}, we begin with an approximation of the term in \eqref{e:hahp}. Lipschitz continuity of $\hah''$, established in \Lemma{t:RBMACOE}, gives
\[
  \hah'(W(t)) =   \hah'(\xi\cdot q + I(t)+ \xi\cdot A(t))
  		 =   \hah'(\xi\cdot q ) +  \hah''(\xi\cdot q)[ I(t)+ \xi\cdot A(t))] +O(1)
\]
Moreover, since $I(t)=0$ for $W(t)\ge -\thresh$ we have  $\hah''(W(t)) I(t) = O(1)$, and hence also
\[
\begin{aligned}
\hah''(\xi\cdot q) I(t) & =\hah''(W(t)) I(t) + O(1)=O(1).  
%\\
%\text{\it and}\quad \hah''(\xi\cdot q)\barIdle(q) &  =O(1).  
\end{aligned}
\]
These two approximations imply that
\[
  \hah'(W(t)) =     \hah'(\xi\cdot q ) +  \hah''(\xi\cdot q)  \xi\cdot A(t) +O(1)
\]
and hence the conditional expectation \eqref{e:hahp} admits the bound,
\[
\begin{aligned}
 \Expect[  & \hah'(W(t)) [W(t+1) - W(t) ] \mid Q(t)=q] 
 \\
 &= \Expect[ \{ \hah'(\xi\cdot q ) +  \hah''(\xi\cdot q)  \xi\cdot A(t) \} [I(t) +\xi\cdot A(t+1) ] \mid Q(t)=q] +O(1)
  \\ 
  &=   \{ \hah'(\xi\cdot q ) +  \hah''(\xi\cdot q) (-\delta) \} [\barIdle(q) -\delta ]   +O(1)
  \\
  &=  \hah'(\xi\cdot q )[\barIdle(q) -\delta ]  +O(1)
\end{aligned}
 \]
where the last bound used $\hah''(w)= O( \delta^{-1})$, which is 
also given in \Lemma{t:RBMACOE}.

Applying Lipschitz continuity of $\hah''$ once more gives an approximation for \eqref{e:hahpp}:
\[
\begin{aligned}
\Expect[  \hah''(W(t)) [W(t+1) - &W(t) ]^2 \mid Q(t)=q]  
\\
  &=  \Expect[  \hah''(\xi\cdot q) [\xi\cdot A(t+1) + I(t) ]^2 \mid Q(t)=q]   +O(1)
  \\
  &=  \Expect[  \hah''(\xi\cdot q) [\xi\cdot A(t+1)  +\delta ]^2 \mid Q(t)=q]   +O(1)   \\
  &=    \hah''(\xi\cdot q) \sigma^2_\Delta +O(1) 
\end{aligned} 
\]
Combining these approximations for \eqref{e:hahp} and \eqref{e:hahpp} gives
\eqref{e:cmwAapp}.

To see that \eqref{e:nabhc} is approximately equal to $B_c(q)$ we begin with the bound,
\begin{equation}
\begin{aligned}
 \Expect[  \nabla  h_c\, (X(t))  & \cdot (U(t) +\alpha)   \mid Q(t) \bigr]  
 \\
& =
  \Expect[     h_c (X(t+1)) -h_c (X(t))      \mid Q(t) \bigr]  +O(1)
 \\ 
 &=2\kappa \zeta(q)\Expect\bigl[    c(\tilX(t+1)) - c(\tilX(t))   \mid Q(t) \bigr] 
 \\
 & \quad -2\kappa  \zeta(q) \Expect\bigl[   
  \barc(\tilW(t+1)) - \barc(\tilW(t))      
  \mid Q(t) \bigr]  +O(1)
\end{aligned}
\label{e:preBbdd}
\end{equation}
where $\zeta(q) = c(q) - \barc(\xi\cdot q) $ was introduced in \eqref{e:barIdle-zeta}.
The term involving workload is bounded using convexity of $\barc$:
\[
\barc(\tilW(t+1)) \ge \barc(\tilW(t)) + \delta_c [\tilW(t+1) - \tilW(t) ]
\]
where $\delta_c$ is any sub-gradient of $\barc$ at $\tilW(t)$.  This leads to the pair of bounds,
\[ 
\Expect[ 
\barc(\tilW(t+1)) - \barc(\tilW(t))  \mid X(t) \bigr] \ge  
\begin{cases}     - \delta \barc_+  &  W(t) \ge 0
\\ 
- I(t) \barc_+  &  W(t) \le 0
\end{cases}
\]
Summing over both cases gives the upper bound,
\[ 
-
\Expect[ 
\barc(\tilW(t+1)) - \barc(\tilW(t))  \mid Q(t) \bigr] \le
\Expect[ \barc_-   I(t)  +\barc_+\delta
  \mid Q(t)=q \bigr]    
\]
This combined with \eqref{e:preBbdd}
 completes the proof that the term \eqref{e:nabhc} is equal to $B_c(q)+O(1)$, which completes the proof of the lemma. 
 \end{proof}

We next establish the implication 
\eqref{e:EnforceIdle} $\Longrightarrow$ \eqref{e:Abdd}:

\begin{lemma}
\label{t:e:EnforceIdle}
If $\barIdle(q) $ satisfies \eqref{e:EnforceIdle} then  \eqref{e:Abdd}
holds:
$
B(q)  \le  -\barc(\xi\cdot q) +  O(1)$.
\end{lemma}

\begin{proof}
Recall the approximation
\eqref{e:cmwAapp}:
\[
B(q) +\haeta^{**} = 
 \Expect[\hah(W(t+1))  -  \hah(W(t))  \mid Q(t)=q]   +O(1)
\]
\Proposition{t:CBMCRWneg}~(i) implies that under   \eqref{e:EnforceIdle} 
\[
 \Expect[\hah(W(t+1))  -   \hah(W(t))  \mid Q(t)=q]   
 =  -\barc(\xi\cdot q)  + \haeta^{**}    + O(1)
\]
These two approximations imply
the desired conclusion that $B(q)   =    -\barc(\xi\cdot q) +O(1)$.
\end{proof}
%
%The bound \eqref{e:EnforceIdle} implies the following bounds on mean-idleness:
%\begin{equation}
%\barIdle(q)  \begin{cases}
%\ge (\delta+\delta_+  ) \quad  & \text{\it if} \quad \xi\cdot q < -\thresh -1
%\\
%\le (\delta+\delta_+  ) \quad
%& \text{\it if} \quad -\thresh -1\le \xi\cdot q < -\thresh +1
%\\
%0 & \text{\it else}
%\end{cases}
%\label{e:SuffIdle}
%\end{equation} 
%This follows immediately from  writing $ \ind\{W(t)<-\thresh\} =  \ind\{\xi\cdot q +\xi\cdot A(t)<-\thresh\} $,  and the fact that $|\xi\cdot A(t)|\le 1$.

\section{Step 2:  Construction of randomized policy}

\subsection{Setting parameters in   $\tilde\bfmx$}
\label{s:e:tilWdrift} 

The following analysis is adapted from Prop.~2.7 of \cite{mey09a}. 

Consider  the general Taylor series approximation for a function $f\colon\Re\to\Re$   that is twice continuously differentiable ($C^2$).  The following second-order version of the Mean Value Theorem holds:  For any $r,r^+\in\Re$,
\[
f(r^+) = f(r) + f'(r)(r^+-r) + \half f'' (\barr)(r^+-r)^2
\]
where $\barr$ lies on the interval with extreme points $r$ and $r^+$.   
The derivatives of the function $f(r) = r + \beta( e^{-r/\beta} - 1)$ are
\[
f'(r)=1- e^{-r/\beta},\qquad f''(r) =  e^{-r/\beta}/\beta 
\]
We take $r=Q_i(t)$ and $r^+=X_i(t+1)$, so that  $-1\le r^+-r\le 2$,
$(r^+-r)^2\le 4$,  $f'' (\barr)\le f'' (( r-1)_+) $, and consequently 
\[
f(r^+) \le f(r) + f'(r)(r^+-r) + 2 f'' (( r-1)_+) 
\]
Using the notation  $\tilX_i(t+1) = f(X_i(t+1))$, $\tilQ_i(t) = f(Q_i(t))$,
this becomes
\[
\tilX_i(t+1) \le \tilQ_i(t) 
	+  [1- e^{-q_i/\beta}][X_i(t+1)-Q_i(t)] 
	+ \frac{2}{\beta} e^{-(q_i-1)_+/\beta}\, ,
\]
with $q=Q(t)$.
Moreover, the function $f$ is convex, so that
\[
\tilX_i(t) = f(Q_i(t)+A_i(t)) \ge  f(Q_i(t)) +f'(Q_i(t)) A_i(t) = \tilQ_i(t) + [1- e^{-q_i/\beta}] A_i(t)
\] 
Combining these bounds gives,
\begin{equation}
\tilX_i(t+1) -\tilX_i(t) \le 
	  [1- e^{-q_i/\beta}][X_i(t+1)-X_i(t)]  
	+ \frac{2}{\beta} e^{-(q_i-1)_+/\beta}\, .
\label{e:tilX_TS}
\end{equation}

With these preliminaries we obtain the following corollaries to \Lemma{t:RandomizedHeavyDrift}.  We begin with the implication of \eqref{e:RandomizedHeavyDriftA}:

\begin{lemma}
\label{t:nullDrift}
Suppose that $q=Q(t)$ satisfies one of two conditions:   Either $q_i=0$, or $q_i\ge 1$ and
the zero-drift condition \eqref{e:RandomizedHeavyDriftA} holds.  Then
\begin{equation}
\Expect\bigl[ \tilX_i(t+1)  -\tilX_i(t) \mid Q(t)\bigr]   \le    \frac{2}{\beta}  
\label{e:nullDrift}
\end{equation}
\end{lemma}

\begin{proof}
The bound is obtained on taking conditional expectations of each side of \eqref{e:tilX_TS},
\begin{equation}
\begin{aligned}
\Expect\bigl[ \tilX_i(t+1) -\tilX_i(t) \mid Q(t)=q\bigr] 
&  \le   [1- e^{-q_i/\beta}]\Expect\bigl[  X_i(t+1)-X_i(t)
  \mid Q(t) =q \bigr]  
  \\
&  \qquad	+ \frac{2}{\beta} e^{-(q_i-1)_+/\beta}
\end{aligned}
\label{e:tilX_TSa}
\end{equation}   
If $q_i=0$ then $[1- e^{-q_i/\beta}]=0$, giving \eqref{e:nullDrift}.   Otherwise 
$[1- e^{-q_i/\beta}]>0$, but we always have
\[
 \Expect [  X_i(t+1)-X_i(t) \mid Q(t)=q  ]  =\Expect [  Q_i(t+1)-Q_i(t) \mid Q(t)=q  ]  
\]
The right hand side is non-positive under   \eqref{e:RandomizedHeavyDriftA}, so that \eqref{e:tilX_TSa}
again implies \eqref{e:nullDrift}.  
\end{proof}

The value of $\beta$ can now be set based on the value of $\epsy_0>0$ appearing in
\eqref{e:RandomizedHeavyDriftB}.  Throughout the remainder of the appendix it is chosen so that the following bound holds:
\begin{equation}
|c|   \frac{2}{\beta} \le \frac{\epsy_0}{4}
\label{e:beta}
\end{equation}
where $|c|=\sum c_i$.   Then, using \eqref{e:nullDrift}, it follows that 
the sum of positive drift from all ``null buffers'' is at most $\epsy_0/4$.

With the values of $\epsy_0>0$ and $\beta>0$ fixed,  choose $\barq\ge 1$
so that the following bound holds:
\begin{equation} 
\Bigl(
- [1- e^{-\barq/\beta}] \epsy_0  + \frac{2}{\beta} e^{-(\barq-1)/\beta} \Bigr)  \le-\half \epsy_0 
\label{e:barq}
\end{equation}

\begin{lemma}
\label{t:strictDrift}
Suppose that the pair of bounds in \eqref{e:RandomizedHeavyDriftB} hold,  with
 $q_k\ge \barq$,  $q_m\ge 1$.  Then,
\begin{equation} 
\Expect\bigl[     \tilX_k (t+1)    - \tilX_k (t)  \mid Q(t)\Bigr]  \le     -\half \epsy_0
\label{e:RandomizedHeavyDriftC}
\end{equation} 
\begin{equation} 
\begin{aligned}
\Expect\bigl[   c_k \tilX_k (t+1)   +c_m \tilX_m (t+1) 
- \bigl( c_k \tilX_k (t) & +c_m \tilX_m (t) \bigr)
  \mid Q(t)\bigr] 
  \\
	&\le   c_m \frac{2}{\beta}  -\half \epsy_0  
\end{aligned}
\label{e:RandomizedHeavyDriftD}
\end{equation}  
\end{lemma}

\begin{proof}
The Taylor series bound \eqref{e:tilX_TSa} gives,
\[
\Expect\bigl[ \tilX_k(t+1) -\tilX_k(t) \mid Q(t)=q\bigr] 
  \le  [1- e^{-q_k/\beta}] (-\epsy_0 ) + \frac{2}{\beta} e^{-(q_k-1)/\beta}
\] 
This combined with \eqref{e:barq} implies \eqref{e:RandomizedHeavyDriftC} when $q_k\ge \barq$.

The bound \eqref{e:RandomizedHeavyDriftD} also follows from \eqref{e:tilX_TSa}.  The additional term $c_m /2\beta$ appears because no lower bound has been imposed on $q_m$.
\end{proof}

%
%\subsubsection*{Proof of \eqref{e:Bbdd},  and \eqref{e:tilWdrift}}
%
%We first demonstrate that the bound \eqref{e:tilWdrift} implies \eqref{e:Bbdd}.   We first apply Jensen's inquality,
% 
%
%The previous two lemmas give,
%\begin{equation}
%\Expect\bigl[ c(\tilX(t+1)) - c(\tilX(t)) \mid Q(t)=q\bigr] \le -\epsy_1   
%\label{e:tilWdriftA}
%\end{equation}
% 
% ...
% 
%Define $g(w) = \barc(\tilw)$, $w\in\Re$,  where $\tilw$ is defined in \eqref{e:tilw}.   This is a convex function, so that  by Jensen's inequality,
%\[
%\begin{aligned}
%\Expect\bigl[   \barc(\tilW(t+1)) \mid Q(t)=q\bigr] & = 
%\Expect\bigl[  g(W(t+1)) \mid Q(t)=q\bigr] 
%\\
%&\ge g\bigl( \Expect\bigl[  W(t+1) \mid Q(t)=q\bigr] \bigr)
%\\
%& =  g\bigl( \Expect\bigl[  W(t) + I(t) -\delta \mid Q(t)=q\bigr] \bigr)
%\end{aligned}
%\] 
%
%
%\begin{equation}
%\Expect\bigl[   \barc(\tilW(t+1)) \mid Q(t)=q\bigr]
%\ge  
%-\epsy_1  + \Expect\bigl[  \barc(\tilW(t)) \mid Q(t)=q\bigr]
%\label{e:tilWdriftB}
%\end{equation}
\subsection{Proof of \Lemma{t:RandomizedHeavyDrift}}

%The main difficulty in translating the previous randomized policy into a policy that depends on $X$ is in the fact that we don't have anymore the information who just arrived - we have to apply the some matching rule to everyone that is in the system. First idea: consider the non-empty queues and apply the previous matching policy. Problem: the order matters.

{\it Step I: decomposition into two connected components and the basic network flows.}

%A search for a policy that allows no cross-matching corresponds to a search for a policy in a new matching graph without the cross-matching arcs. 
A nonidling policy %for the original matching graph 
corresponds to a policy in a new matching graph without the arcs between $S$ and $D^c$. 
This cuts the matching graph into two subgraphs, $\clG_1 = \clG_{D_1 \cup S_1}$ and $\clG_2 =\clG_{D_2 \cup S_2}$, with $D_1 = D$, $S_1 = \Sply(D_1)$, $D_2 = D^c$, and $S_2 = S_1^c$.
%From assumption (A1) it  follows that both subgraphs are connected. 
%{\color{red}
\begin{lemma}
Under assumption (A1), both subgraphs $\clG_1$ and $\clG_2$ are connected. 
\end{lemma}
\begin{proof}
The proof is by contradiction: suppose that $\clG_1$ is not connected. By definition of $S_1 = \Dmd(D_1)$, for any $s\in S_1$ there is an arc to some $d \in D_1$. Therefore, $\clG_1$ not being connected implies that $D_1$ can be decomposed as $D_1 = D'_1 \cup D''_1$, $D'_1, D''_1 \neq \emptyset$, with no path 
in $\clG_1$ between $D'_1$ and $D''_1$. Set $S'_1 = \Sply(D'_1)$ and $S''_1 = \Sply(D''_1)$. 
Then, 
%\begin{alignat*}{2}
%\xi^{D_1}\cdot \alpha &=& -\delta \\
%\xi^{D'_1}\cdot \alpha &=& -\delta' \\
%\xi^{D''_1}\cdot \alpha &=& -\delta'' 
%\end{alignat*}
\[
\xi^{D_1}\cdot \alpha  =  -\delta,  \quad
\xi^{D'_1}\cdot \alpha =  -\delta', \quad \textrm{and} \quad
\xi^{D''_1}\cdot \alpha = -\delta'',
\]
with $\delta', \delta''>0$ and $\delta = \delta' + \delta''$. 
These implications violate Assumption (A1). 
%Thus (A1) is not satisfied for $D'_1$ and $D''_1$. 

For $\clG_2$ the arguments are symmetrical (with $S_2$ playing the role of $D_1$). 
 \end{proof}
%}

Let $\clN_1$ and $\clN_2$ be the corresponding networks as defined by (\ref{eq-dg}). 
%Then, by Lemma \ref{t:MCMF}, there is a strictly positive flow $F_1$ of value $\alpha_{D_1}$ for  network %$\clN_1$,
%and a strictly positive flow $F_2$ of value $\alpha_{S_2}$ for network $\clN_2$. 
%
The first step is to establish a slightly stronger version of Lemma \ref{t:MCMF}. In what follows, we only consider $\clN_1$, the arguments for $\clN_2$ are symmetrical. 

%\begin{lemma}
%Under assumption (A1), there is a flow $F_1$ of value $\alpha_{D_1}$ for  network $\clN_1$,
%and a  flow $F_2$ of value $\alpha_{S_2}$ for network $\clN_2$ such that
%\begin{equation}
%\label{eq:gamma}
%\gamma = \min \left \{ \min_{e \in \IEdge_1} F_1(e), \min_{e \in \IEdge_2} F_2(e) \right \} \geq \frac{\udelta}{|\IEdge|}. % > 0.
%\end{equation}
%%Under assumption (A1), $\gamma$ does not depend on $\delta$. 
%\end{lemma} 

\begin{lemma}
Under Assumptions (A1)-(A3), there is a flow $F_1$ of value $\alpha_{D_1}$ for  network $\clN_1$
such that
\begin{equation}
\label{eq:gamma}
\gamma = \min_{e \in \IEdge_1} F_1(e) \geq \frac{1}{|\IEdge_1|}
\min_{k\in D_1 \cup S_1} \inf_{\delta\in [0,\maxdelta]}\alpha^{\delta}(k) > 0.
\end{equation}
%Under assumption (A1), $\gamma$ does not depend on $\delta$. 
\end{lemma} 

%\item[(A1)]
%For one set $D\subsetneq \IDmd$ we have $\xi^D\cdot \alpha^\delta =-\delta$,   where $\alpha^\delta$ denotes the mean of $A^\delta(t) $.
%
%Moreover, there is a fixed constant $\udelta>0$ such that  $\xi^{D'}\cdot \alpha^\delta \le -\udelta$ for any $D'\subsetneq \IDmd$, $D'\neq D$, and $\delta\in [0,\maxdelta]$.

\begin{proof}
The proof follows similar arguments as in \cite[Lemma 3.2]{busgupmai13}. 
Suppose that Assumption (A1)-(A3)  are
satisfied. 
Denote by $\Dmd_1(s) = \{d \in \Dmd_1 \; : \; (d,s) \in \IEdge_1\}$ and 
$\Sply_1(d) = \{s \in \Sply_1 \; : \; (d,s) \in \IEdge_1\}$. Note that $\Sply_1(d) = \Sply(d), \; d\in \Dmd_1$, by the definition of $S_1$. 

Fix 
$$
\nu = \frac{1}{|\IEdge_1|}    \min \left \{ \frac{\udelta}{2 } , {\min_{k\in D_1 \cup S_1} \inf_{\delta\in [0,\maxdelta]}\alpha^{\delta}(k)}\right \}
$$
which under Assumption (A3) is strictly positive. As $D_1 \subsetneq \IDmd$, $\alpha(D_1) < 1$, and hence $\nu <1/{|\IEdge_1|}$. Consider the function $F_{\nu}: \IEdge_1 \rightarrow \R_+$ defined by 
\[
F_{\nu}(x,y)=\begin{cases}
\nu & \mbox{for } (x,y)=(d,s)\in \IEdge_1 \\
|\Sply_1(d)|\ \nu  & \mbox{for } (x,y)=(a,d), \; d \in D_1 \\
|\Dmd_1(s)|\ \nu  & \mbox{for } (x,y)=(s,f), \; s\in S_1\:.
\end{cases}
\]
By construction, $F_{\nu}$ is a flow for $\clN_1$. Set 
\begin{alignat*}{2}
%\label{eq-tilde}
\widetilde{\alpha} (d) = \frac{ \alpha(d) - |\Sply_1(d)|\nu }{1- |\IEdge_1|\nu}, &\quad d\in D_1 \\
\widetilde{\alpha} (s) = \frac{ \alpha(s) - |\Dmd_1(s)|\nu }{1- |\IEdge_1|\nu}, &\quad s\in S_1 \:.
\end{alignat*}
As $\nu \leq  \min \left\{\min_{d\in D_1} \frac{\alpha(d)}{|\Sply_1(d)|}, \min_{s\in S_1} \frac{\alpha(s)}{|\Dmd_1(s)|}\right \}$, $\widetilde{\alpha}(k) \geq 0$ for $k \in D_1 \cup S_1$.

%Set $\nu = \frac{1}{|\IEdge_1|}\min_{k \in D_1 \cup S_1} \alpha(k) < \frac{ 1}{|\IEdge_1|} $. 
For any $D' \subsetneq D_1$, $\xi^{D'} \leq \xi^{\hat{D'}}$, with $\hat{D'} = \cup \{ D'' : \Sply(D'') = \Sply(D')\}$. Thus 
\begin{alignat*}{2}
\xi^{D'} \cdot  \widetilde{\alpha} &\leq  \xi^{\hat{D'}} \cdot \widetilde{\alpha} = \frac{ 1}{1- |\IEdge_1|\nu}\left (\xi^{\hat{D'}}\cdot \alpha - \nu (\sum_{d\in \hat{D'}} |\Sply_1(d)|- \sum_{s\in \Sply(D') } |\Dmd_1(s)| ) \right ) \\
& \leq   \frac{ 1}{1- |\IEdge_1|\nu}\left (- \udelta + \nu \sum_{s\in \Sply(D') } |\Dmd_1(s)| \right ) 
\leq   \frac{ 1}{1- |\IEdge_1|\nu}\left (- \udelta + \nu |\IEdge_1| \right ) < 0\, . 
\end{alignat*}

Consider the directed graph $\clN_1$, %see \eref{eq-dg}, 
with new capacities on the
demand and supply arcs defined by $\widetilde{\alpha}$. 

The above shows that {\sc NCond} (\ref{e:NCond}) are satisfied for $\clN_1$ at $\delta = 0$. 
By Assumption (A2) this is still satisfied for any $S' \subsetneq S_1$ for a small enough interval $[0,\maxdelta]$. 
\spm{I wish this could be made clear at the top -- that we will justify a continuity argument.  No changes needed now.}
By applying Lemma \ref{t:MCMF}, there exists a flow
$\widetilde{F_1}: \IEdge_1 \rightarrow \R_+$ of value $\alpha_{D_1}$. Define 
\[
F_1: \IEdge_1 \rightarrow \R_+, \qquad F_1 =  F_{\nu} + (1- |\IEdge_1|\nu)
\widetilde{F_1} \:.
\]
By construction $F_1$ is a flow for the graph $\clN_1$
with the original capacity constraints $\alpha$. The value of $F_1$ is $\alpha_{D_1}$ and it
satisfies $F_1(e)\geq \gamma$ for all $e\in \IEdge_1$. 
%This completes the proof.
\end{proof}

The corresponding basic randomized policy is given by equations (\ref{e:pi}) and (\ref{e:pj}) for the two subgraphs. In what follows, we  concentrate only on the subgraph $\clG_1$ (the 
analysis of the subgraph $\clG_2$ is symmetrical). 

For a state $q \in\dstate$ satisfying $c(q) < \barc(\xi\cdot q) + \minc$, this basic policy satisfies (\ref{e:RandomizedHeavyDriftA}) of  \Lemma{t:RandomizedHeavyDrift}. This follows directly from \Lemma{t:ZeroDrift}. 
For a state $q \in\dstate$ satisfying $c(q)\ge \barc(\xi\cdot q) + \minc$, 
 we  slightly modify this basic randomized policy to get a strictly negative drift. 

\medskip 

\noindent
{\it Step II: modified network flow.}

%{\it Step IV: translating a network flow into a randomized policy.}

We assume in the following that $\xi\cdot q\ge 0$. 
%\ana{TODO: justify why WLOG.}
In the case $\xi\cdot q < 0$, the arguments are similar. 
%we can apply similar arguments as what follows for the subgraph that consists of $D^c$ and $S^c$.  
%\ana{More notation:}

Based on the definitions (\ref{e:barcpl}, \ref{e:barcpl-barc}),  a state $q \in\dstate$ satisfying $c(q)\ge \barc(\xi\cdot q) + \minc$ satisfies at least one of the following: 
\begin{eqnarray}
\label{eq:case1}
\sum_{i\in D_1 } q_i c_i + \sum_{j\in S_1 } q_j c_j \geq (\xi\cdot q)\barcpD + \minc/2; \\
\label{eq:case2}
\sum_{i\in D_2 } q_i c_i + \sum_{j\in S_2} q_j c_j \geq (\xi\cdot q)\barcpS + \minc/2.
\end{eqnarray}  
We   consider only the first case;
 the second is symmetrical.

\begin{lemma}
\label{t:GapDriftCase1}
Suppose that \eqref{eq:case1} 
holds, with $\minc$ is given in \eqref{e:minc}.

Then, at least one of the following is satisfied: 
\begin{enumerate}[(a)]
\item there is some $j\in S_1 $ such that $q_{j} \geq \barq$; 
\item there is some $i \in D_1 $ such that $c_{i} > \barcpD$ and $q_{i} \geq \barq$.
\end{enumerate} 
\end{lemma}

\begin{proof}
If $q$ does not satisfy either of the two above conditions, then
\begin{eqnarray*}
\sum_{i\in D_1 } q_i c_i + \sum_{j\in S_1 } q_j c_j 
& = & (\xi\cdot q)\barcpD +  \sum_{i\in D_1 } q_i (c_i - \barcpD) + \sum_{j \in S_1 } q_j(c_j + \barcpD ) 
\\ 
& \leq & (\xi\cdot q)\barcpD  +   \barq  \sum_{i\in D_1 }   (c_i - \barcpD) +   \barq \sum_{j \in S_1 }  (c_j + \barcpD ) 
\\ 
& \leq & (\xi\cdot q)\barcpD  + 2\ell\barq  \cmax    
< (\xi\cdot q)\barcpD + \minc/2,
\end{eqnarray*}
where the first inequality uses the fact that $c_i \ge \barcpD$ for each $i\in D_1$,
and in the second we used the bound $|D_1|+|S_1|\le \ell$ and the definition of $\minc$ in \eqref{e:minc}.   The final inequality is in contradiction with \eqref{eq:case1}.
\end{proof}

We next consider cases (a) and (b) of \Lemma{t:GapDriftCase1}
separately.

{\it Case (a):}
In this case there is also some $i\in D_1 $ such that $q_i \geq 1$. 
Indeed, 
\begin{equation}
\label{eq:path1}
\sum_{\ell\in D_1 } q_{\ell} = \xi\cdot q +  \sum_{k\in S_1 } q_k \geq \sum_{k\in S_1 } q_k \ge \barq,
\end{equation} 
as we consider the case $\xi\cdot q \geq 0$.  For this choice of $i$, $j$, we
modify the basic randomized matching policy to increase the matching rate of classes $i$ and $j$ slightly above their arrival rate.

Without loss of generality, we assume that $(i,j) \not\in \clE_1 = \IEdge_{D_1 \cup S_1}$ (otherwise a modified randomized policy can be obtained 
by first matching an item $i$ with an item $j$ and then using the basic randomized policy). 
There is a simple path connecting $j$ to $i$ using edges in $\clE_1$ (since $\clG_1$ is connected). Denote this path by 
\begin{equation}
\label{eq:path1j}
j=j_1, i_1, j_2, i_2, j_3, \ldots, i_{m-1}, j_{m},  i_m = i.
\end{equation} 
Take $0 < \epsy_1 < \gamma$ and 
consider a new network problem $\clN_1'$ in which the capacities of arcs $(a,i)$ and $(j,f)$ are increased by $\epsy_1$. Define a new flow $F_1'$ by adding  $\epsy_1 (1, -1, 1, \ldots, 1)$ along the path (\ref{eq:path1j}):
$$
\begin{aligned}
F_1'(i_k,j_k) &= F_1(i_k, j_k) + \epsy_1, 1 \leq k \leq m
\\
\textrm{\it and } \quad F_1'(i_k,j_{k+1}) &= F_1(i_k, j_{k+1}) - \epsy_1, 1 \leq k \leq m-1.
\end{aligned}
$$
For the other entries, $F_1'(e) = F_1(e)$. 

For $i$ and $j$, we now have
$$
\sum_{s \in S_1} F'_1(i,s) = \alpha_i + \epsy_1, \quad \sum_{d \in D_1} F'_1(d,j) = \alpha_j + \epsy_1. 
$$

The modified matching probability vectors are defined as follows: 
\begin{itemize}
\item For all $d \in D_1$ such that $d \neq i$, compute $\hat{p}^{(d)}$ using (\ref{e:pi}) for the modified flow $F'_1$. 
\item For all $s \in S_1$ such that $s \neq j$, compute $\hat{p}^{(s)}$ using (\ref{e:pj}) for the modified flow $F'_1$. 
\item For $i$ and $j$ use the basic flow $F_1$: $\hat{p}^{(i)} = p^{(i)}$ and $\hat{p}^{(j)} = p^{(j)}$. 
\end{itemize}
A modified randomized policy can now be defined as follows:
\begin{itemize}
\item If the queue $i_1$ is not empty, an item $i_1$ is matched with an item $j$ (recall that $q_j > 0$). 
\item If the queue $j_m$ is not empty, an item $i$ is matched with an item $j_m$ (recall that $q_i > 0$). 
\item New arrivals choose their potential match independently, according to the matching vectors $\hat{p}$. It may be possible that the chosen queue is already empty after the first two steps: in that case, the new arrival is stored in the buffer.
%\spm{March 2016.  this remark in parentheses doesn't help to clarify. Can we delete?
%\\
%(We may have also defined the matching vectors with respect to the new state $q'$ after the first two steps, %but this is not necessary). }
\end{itemize}

Using similar arguments as in the proof of \Lemma{t:ZeroDrift}, we can analyze the drift of all non-empty queues. We distinguish six cases: 
\begin{itemize}
\item 
For queue $i$, $q_i > 0$ and 
%$\Expect\bigl[  A_i(t) \mid Q(t)=q\bigr] =\Expect\bigl[  A_i(t) \bigr] $ and 
$\sum_{s \in \Sply (i)}F'_1(i,s) = \sum_{s\in \Sply (i)}F_1(i,s) + \epsy_1 = \alpha_i + \epsy_1$, thus
$$
\begin{aligned}
\Expect\bigl[  U_i(t) \mid Q(t)=q\bigr] &=\Expect\bigl[ \Expect [ U_i(t) \mid Q(t), A(t) ] \mid Q(t)=q\bigr] 
\\
&\geq \sum_{s \in \Sply (i)}F'_1(i,s) = \alpha_i + \epsy_1, 
\end{aligned}
$$
%thus
and 
\begin{equation}
\label{eq:DriftI}
\begin{aligned}
\Expect\bigl[ Q_i (t+1) \mid Q(t)=q\bigr]
&= 
\Expect\bigl[ Q_i (t) + A_i(t) - U_i(t) \mid Q(t)=q\bigr] 
\\
%\leq q_i+ \alpha_i - \alpha_i  - \epsy_1
&
\leq q_i  -  \epsy_1.
\end{aligned}
\end{equation}
\item
For queue $j$, $q_j \geq \barq$, and $\sum_{d \in D_1 \cap \Dmd(j) } F_1'(d, j) = \sum_{d \in D_1 \cap \Dmd (j) } F_1(d, j) + \epsy_1$, thus
\begin{equation}
\label{eq:DriftJ}
\Expect\bigl[ Q_j (t+1) \mid Q(t)=q\bigr]
\leq q_j + \alpha_j - \sum_{d \in D_1 \cap \Dmd(j) } F_1(d, j)- \epsy_1 \leq q_j + \delta - \epsy_1, 
\end{equation}
where the last inequality follows as in the proof of \Lemma{t:ZeroDrift}. 
%For all $s \in S_1$, 
%$\sum_{d \in D_1 \cap \Dmd(s)} F_1(d,s) \leq \alpha_s$, and 
%$$
%\sum_{s \in S_1} \Bigl (\alpha_s - \sum_{d \in D \cap \Dmd(s)} F(d,s) \Bigr) = \alpha_{S_1} - \alpha_{D_1} = \delta, 
%$$
%In particular, $\alpha_j  - \sum_{d \in D_1 \cap \Dmd(j)} F_1(d,j) \leq \delta.$

\item 
For queue $i_1$, 
%$\sum_{s \in \Sply(i_1)}F'_1(i_1,s) = \sum_{s\in \Sply (i_1)}F_1(i_1,s)$. However, the new flow $F_1'$ is only used %for the new arrivals $s \neq j$. Furthermore, 
if $q_{i_1} > 0$, there is always a match $(i_1, j)$ as defined by step 1 of the modified randomized policy. 
Thus, if $q_{i_1} > 0$, 
$$
\Expect\bigl[ Q_{i_1} (t+1) \mid Q(t)=q\bigr]
\leq q_{i_1} + \alpha_{i_1}  - 1 \leq q_{i_1}.
$$

\item 
For queue $j_m$, if $q_{j_m} > 0$, then there is always a match $(i, j_m)$ as defined by step 2 of the modified randomized policy. Thus, if $q_{j_m} > 0$, 
$$
\Expect\bigl[ Q_{j_m} (t+1) \mid Q(t)=q\bigr]
\leq q_{j_m} + \alpha_{j_m} - 1 \leq q_{j_m} . 
$$

\item 
For queue $d \in D_1 \backslash \{i, i_1\}$, $\sum_{s \in S (d)}F'_1(d,s) = \sum_{s\in S (d)}F_1(d,s)$. 
Thus, if $q_{d} > 0$,
$$
\Expect\bigl[ Q_d (t+1) \mid Q(t)=q\bigr]
\leq q_d.
$$

\item 
For queue $s \in S_1 \backslash \{j, j_m\}$, $\sum_{d \in D_1 \cap \Dmd (s)}F'_1(d,s) = \sum_{d\in D_1 \cap \Dmd(s)}F_1(d,s)$. 
Thus, if $q_{s} > 0$,
$$
\Expect\bigl[ Q_s (t+1) \mid Q(t)=q\bigr]
\leq q_s + \alpha_s - \sum_{d \in D_1 \cap \Dmd (s) } F_1(d, s) \leq q_s + \delta.
$$
%where the last inequality follows using  the same arguments as in the case of queue $j$. 

\end{itemize}

This establishes the conclusions of \Lemma{t:RandomizedHeavyDrift}, with $k= i$, $m = j$, $\epsy_0 = \epsy_1 > 0$, 
and
$\epsy'_0 = (c_i + c_{j})(\epsy_1 - \delta) > 0$,  
for $\delta$ small enough.

\medskip

{\it Case (b):}
In this case,  we will modify the basic flow to increase the matching rate of class $i$ above its arrival rate. At the same time, we will decrease the matching rate of queue $i_0$ (a queue with cost $\barcpD$).

There is a path connecting $i$ to $i_0$ using edges in $\IEdge_1$. Denote this path by 
$$
i=i_1, j_1, i_2, j_2, \ldots, i_{n-1}, j_{n-1}, i_n = i_0.
$$ 
Take $0 < \epsy_2 < \gamma$. 
Consider a new network problem in which the capacity of arc $(a,i)$ is increased and of arc $(a,i_0)$ decreased by $\epsy_2$. Define a new flow $F_1''$ by adding  $\epsy_2 (1, -1, 1, \ldots, -1)$ on the above path from $i$ to $i_0$:
$$
\begin{aligned}
F_1''(i_k,j_k) &= F_1(i_k, j_k) + \epsy_2, 1 \leq k \leq n-1
\\
\textrm{\it and } \quad F_1''(i_k,j_{k-1}) &= F_1(i_k, j_{k-1}) - \epsy_2, 2 \leq k \leq n.
\end{aligned}
$$
For the other entries, $F_1''(e) = F_1(e)$. 

The rest of the proof is now similar. We define a new randomized policy using $F_1''$ for the matching vectors for all queues in $D_1$ and $S_1$, except for queues $i$ and $i_0$ that keep the matching vectors defined using the basic flow $F_1$.  

A modified randomized policy is defined as follows:
\begin{itemize}
\item If the queue $j_{1}$ is not empty, an item $i$ is matched with an item $j_{1}$ (we know that $q_{i} > 0$). 
\item New arrivals are choose their potential match independantly, according to their matching vectors. 
\end{itemize}

This has impact only on queues $i$, $i_0$, $j_1$, and $j_{n-1}$;
 for the others, the drift remains the same as for the basic randomized policy. 
%and the neighbors of $i$ 
%(a supply class $s \in \Sply (i_0)$ will
%get matched at least with the same rate as in the basic case, as $F_1''(i_0,j_{n-1}) < F_1(i_0,j_{n-1})$ and 
%$F_1''(i_0,s) = F_1(i_0,s)$ for all the other $s \in \Sply(i_0)$).
\begin{itemize}
\item For queue $i$, $q_{i} > 0$ and 
$\sum_{s \in \Sply (i)}F''_1(i,s) = \sum_{s\in \Sply (i)}F_1(i,s) + \epsy_2 = \alpha_i + \epsy_2$, thus
$$
\Expect\bigl[ Q_i (t+1) \mid Q(t)=q\bigr]
\leq q_i  -  \epsy_2.
$$

\item For queue $i_0$, 
$\sum_{s \in \Sply (i)}F''_1(i_0,s) = \sum_{s\in \Sply (i_0)}F_1(i_0,s) - \epsy_2 = \alpha_{i_0} - \epsy_2$, thus
if $q_{i_0} > 0$, then 
$$
\Expect\bigl[ Q_{i_0} (t+1) \mid Q(t)=q\bigr]
\leq q_{i_0}  +  \epsy_2.
$$

\item For queue $j_1$,  
if $q_{j_1} > 0$, then there is always a match $(i, j_1)$ as defined by step 1 of the modified randomized policy. Thus, if $q_{j_1} > 0$, 
$$
\Expect\bigl[ Q_{j_1} (t+1) \mid Q(t)=q\bigr]
\leq q_{j_1} + \alpha_{j_1} - 1 \leq q_{j_1}. 
$$

\item For queue $j_{n-1}$, 
\begin{eqnarray*}
\Expect\bigl[  U_{j_{n-1}}(t) \mid Q(t)=q\bigr] & =& \Expect\bigl[ \Expect [ U_{j_{n-1}}(t) \mid Q(t), A(t) ] \mid Q(t)=q\bigr] \\
&\geq & \sum_{d \in D_1 \cap \Dmd (j_{n-1}) \backslash \{i_0\} }F''_1(d,j_{n-1}) + F_1(i_0, j_{n-1})\\
&\geq & \sum_{d \in D_1 \cap \Dmd (j_{n-1})  }F_1(d,j_{n-1}) + \epsy_2,
\end{eqnarray*}
%thus
and 
\[
\begin{aligned} 
\Expect\bigl[ Q_{j_{n-1}}(t+1) \mid &  Q(t)= q\bigr]
\\
&= 
\Expect\bigl[ Q_{j_{n-1}} (t) + A_{j_{n-1}}(t) - U_{j_{n-1}}(t) \mid Q(t)=q\bigr] \\
&\leq  q_{j_{n-1}}+ \alpha_{j_{n-1}} - \sum_{d \in D_1 \cap \Dmd (j_{n-1})  }F_1(d,j_{n-1})  - \epsy_2\\
&\leq  q_{j_{n-1}} + \delta  -  \epsy_2. 
\end{aligned}
\]

\end{itemize}

We again obtain the conclusions of \Lemma{t:RandomizedHeavyDrift}, with $k=i, m=i_0$,  $\epsy_0 = \epsy_2 > 0$,  
and $\epsy'_0 = (\barcpD - c_{i})\epsy_2 > 0$.

\subsection{Drift with idling}

We need a corresponding lemma when idling is required ($W(t)<-\thresh$): 
\begin{lemma}
\label{t:RandomizedHeavyDriftIdle}
Under the assumptions of \Thm{t:hMWao},  
there exist constants  
$\epsy_0>0$,     $\bardelta_0\in (0,\maxdelta)$, 
and $\bar\delta_+>0$,
such that for each $\delta_+\in (0,\bar\delta_+)$
 and $\delta \in [0,\bardelta_0]$, there is
a randomized policy that satisfies the   uniform bounds (i) and (ii) in
\Lemma{t:RandomizedHeavyDrift}.  
%\spmold{9/13:  NOTE!  Let's discuss to be sure this is clear.}
In addition, 
the corresponding workload satisfies, 
\begin{equation}
\Expect[I(t) \mid Q(t)] -\delta =
\Expect[W(t+1) - W(t) \mid Q(t)] = \delta_+  
\label{e:plusdrift}
\end{equation}
\end{lemma}

\begin{proof}
Consider the basic randomized policy as in Step II of the proof of \Lemma{t:RandomizedHeavyDrift}, obtained using a strictly positive flow  $F_1$  of value $\alpha_D$ for the network flow problem $\clN_1$ defined in Step~I of the proof of \Lemma{t:RandomizedHeavyDrift}.  
Set $\gamma = \min_{e \in \IEdge_1} F_1(e) > 0$,
as in (\ref{eq:gamma}). 

Most of the arguments are the same as in the proof of \Lemma{t:RandomizedHeavyDrift}. We highlight only the differences in what follows. One main difference is that, to ensure idling in (\ref{e:plusdrift}), we now need cross-matchings (matchings of items in $S_1$ with demand items in $D_2 = D_1^c$). The other is the fact that  Step~II of the proof of  Lemma  \ref{t:RandomizedHeavyDrift} was written assuming $\xi\cdot q \geq 0$. % (and stating that the other case is similar). 
Here we have $\xi\cdot q < -\thresh < 0$, so we first start by explaining the similarity with Step~II of  the proof of  Lemma  \ref{t:RandomizedHeavyDrift}. 

%Let $d \in \argmin_{i\in  D^c} c_i$ and $s \in \argmin_{j\in S } c_j$.
%The state $q'$ with $q'_d = q'_s = \xi\cdot q$ and $q'_k = 0, k \not\in \{d,s\}$ satisfies $\xi\cdot q' = \xi\cdot q$ and 
%$c(q') = \barc(\xi\cdot q)$. 

A state $q \in\dstate$ satisfying $c(q)\ge \barc(\xi\cdot q) + \minc$ satisfies at least one of the following:
\begin{eqnarray}
\label{eq:case3}
\sum_{i\in D_1 } q_i c_i + \sum_{j\in S_1 } q_j c_j \geq  - (\xi\cdot q)\barcmS + \minc/2; 
\\
\label{eq:case4}
\sum_{i\in D_2} q_i c_i  + \sum_{j\in S_2} q_j c_j \geq - (\xi\cdot q)\barcmD+ \minc/2.
\end{eqnarray}  
We will consider the first case;
the second is symmetrical. 

\Lemma{t:GapDriftCase1} can be extended to the present setting:
For any $\minc > 0$ defined in \eqref{e:minc}, at least one of the following is satisfied
\begin{enumerate}[(a)]
\item there is some $i\in D_1 $ such that $q_{i} \geq \barq$; 
\item there is some $j \in S_1 $ such that $c_{j} > \barcmS$ and $q_{j} \geq \barq$.
\end{enumerate}

\medskip

{\it Case (a):}
In this case, under the assumption on $\minc$, there is also some $j\in S_1 $ such that $q_j \geq 1$
(for the details see the equivalent step for the case $\xi\cdot q \geq 0$ in the proof of \Lemma{t:RandomizedHeavyDrift}). 

\medskip

As in the proof of \Lemma{t:RandomizedHeavyDrift}, we modify the basic randomized matching policy to increase the matching rate of classes $i$ and $j$ slightly above their arrival rate.

Without loss of generality, we assume that $(i,j) \not\in \IEdge_1$ (otherwise a modified randomized policy can be obtained 
by first matching an item $i$ with an item $j$ and then using the basic randomized policy). 

As in Step~II of the proof of \Lemma{t:RandomizedHeavyDrift}, there is a path connecting $j$ to $i$ using edges in $\IEdge_1$ (since $\clG_1$ is connected). Denote this path by 
$$
j=j_1 i_1 j_2 i_2 j_3 \ldots i_{m-1} j_{m} i_m = i.
$$ 
Furthermore, under Assumption (A3), 
there exists $s' \in S_1 $ and $d' \in D_2$ such that 
\begin{equation}
 \label{e:Idle}
\Prob\{A^\delta_{s'}(t) \ge 1 \ \text{\it and} \ A^\delta_{d'}(t) \ge 1\} \ge \pcm ,\qquad  0\le \delta\le \maxdelta.
\end{equation} 

We assume that in this case, the newly arrived items $d'$ and $s'$ are matched with some probability $\epsy_2$. 
We need to compensate for the decrease of the matching rate of class $s'$ available for items in $D_1$. We do this by constructing a new path from $j$ to $s'$ (such a path exists since $\clG_1$ is connected). 
Denote this path by 
$$j=j'_1 i'_1 j'_2 i'_2 \ldots j'_{n-1} i'_{n-1} j'_n = s'.$$ 

Take $\epsy_1$ and $\epsy_2 =  (\delta + \delta_+)/{\pcm}$ such that  $0 < \epsy_1 + \epsy_2\pcm  < \gamma$, where $\gamma$ is defined by (\ref{eq:gamma}). 
Consider a new network problem $\clN_1'$ in which the capacity of arc $(a,i)$ is increased by $\epsy_1$ and $(j,f)$ is increased by $\epsy_1 + \epsy_2\pcm$. Define a new flow $F_1'$ in two steps: 
\begin{itemize}
\item Add  $\epsy_1 (1, -1, 1, \ldots, 1)$ on a path from $j$ to $i$ to the flow $F_1$, 
$$
\begin{aligned}
\hat{F_1}(i_k,j_k) &= F_1(i_k, j_k) + \epsy_1, 1 \leq k \leq m
\\
\textrm{ \it and }  \quad\hat{F_1}(i_k,j_{k+1}) &= F_1(i_k, j_{k+1}) - \epsy_1, 1 \leq k \leq m-1.
\end{aligned}
$$
For the other entries, $\hat{F_1}(e) = F_1(e)$.
\item 
Add  $\epsy_2\pcm  (1, -1, 1, \ldots, -1)$ on a path from $j$ to $s$ to the flow $\hat{F_1}$, 
$$
\begin{aligned}
F_1'(i'_k,j'_k) &=\hat{F_1}(i'_k, j'_k) + \epsy_2\pcm , 1 \leq k \leq n-1
\\
\textrm{\it and } \quad F_1'(i'_{k-1},j'_{k})& = \hat{F_1}(i'_{k-1}, j'_{k}) - \epsy_2\pcm , 2 \leq k \leq n.
\end{aligned}
$$
For the other entries, $F_1'(e) = \hat{F_1}(e)$. 
\end{itemize}

The modified matching probability vectors are defined as follows: 
\begin{itemize}
\item For all $d \in D_1$ such that $d \neq i$, compute $\hat{p}^{(d)}$ using (\ref{e:pi}) for the modified flow $F'_1$. 
\item For all $s \in S_1$ such that 
%$s\neq j$, 
$s \not\in \{j, s'\}$, 
compute $\hat{p}^{(s')}$ using (\ref{e:pj}) for the modified flow $F'_1$. 
\item For $s'$, compute $\hat{p}^{(s)}$ using (\ref{e:pj}) for the modified flow $F'_1$ and $\alpha'_{s'} = \alpha_{s'} - \epsy_2\pcm$. 
\item For $i$ and $j$ use the basic flow $F_1$: $\hat{p}^{(i)} = p^{(i)}$ and $\hat{p}^{(j)} = p^{(j)}$. 
\end{itemize}
A modified randomized policy can now be defined as follows:
\begin{itemize}
\item 
%with probability $\epsy_1$ and independently of everything else, 
An item $i$ is matched with an item $j_m$, if $q_{j_m} > 0$.  % (we know that $q_i > 0$). 
\item Queue $j$ is used in both paths, thus the first step is to determine the path to be considered:
\begin{itemize} 
\item With probability $\frac{\epsy_1}{\epsy_1 + \epsy_2 \pcm}$,   an item $j$ is matched with an item $i_1$, if $q_{i_1} > 0$. %(otherwise no match is made at this step).  
\item With probability $\frac{\epsy_2 \pcm}{\epsy_1 + \epsy_2 \pcm}$, an item $j$ is matched with an item $i'_1$,  if $q_{i'_1} > 0$.  %(we know that $q_j > 0$). 
\end{itemize}
\item If the new arrivals are $d'$ and $s'$, then with probability $\epsy_2$ they are matched together. Otherwise $s'$ considered as an usual new arrival for the network $\clN_1$, and it chooses its match according to  $\hat{p}^{(s')}$. 
\item New arrivals choose their potential match independantly, according to the matching vectors $\hat{p}$. It may be possible that the chosen queue is already empty after the first two steps. In that case the new arrival is stored in the buffer. 
%(We may have also defined the matching vectors with respect to the new state $q'$ after the first two steps, %but this is not necessary). 
\end{itemize}
Then, using similar arguments as in the proof of \Lemma{t:RandomizedHeavyDrift}, the conclusions of the lemma hold for 
%$\epsy_0 = c_i\epsy_1 + c_{j}(\epsy_1 + \epsy_2\pcm )>0$
$k=i$, $m=j$, $\epsy_0 = \epsy_1$, $\epsy'_0 = c_i \epsy_1 + c_j (\epsy_1 + \epsy_2 \pcm - \delta) > 0$, %($\epsy'_0 >0$,  for $\delta$ small enough), 
for $\delta$ small enough.
%\spmold{9/13:  I can choose $\epsy_2$ (subject to an upper bound), so I have flexibility with $\delta_+$}

\medskip

{\it Case (b):}
In this case,  we modify the basic flow to increase the matching rate of class $j$ above its arrival rate. At the same time, we decrease the matching rate of $j_0$ (a queue with $c_{j_0} = \barcmS$), using a path connecting $j$ to $j_0$ using edges in $\IEdge_1$: 
$$j=j_1 i_1 j_2 i_2 \ldots j_{n-1} i_{n-1} j_n = j_0.$$ 

As in case (a), 
the newly arrived items $d'$ and $s'$ are matched with some probability $\epsy_4$. 
To compensate for the decrease of the matching rate of class $s'$ available for items in $D_1$, we use a path from $j$ to $s'$:  
$$j=j''_1 i''_1 j''_2 i''_2 \ldots j''_{m-1} i''_{m-1} j''_m = s'.$$ 

Take $\epsy_3$ and $\epsy_4=  (\delta+\delta_+)/\pcm$ such that $0 < \epsy_3 + \epsy_4\pcm  < \gamma$, where $\gamma$ is defined by (\ref{eq:gamma}). 
Consider a new network problem in which the capacity of arc $(j,f)$ is increased by $\epsy_3 + \epsy_4\pcm $ and the capacity of $(j_0,f)$ decreased by $\epsy_3$. Define a new flow $F_1''$ in two steps: 
\begin{itemize}
\item Add  $\epsy_3 (1, -1, 1, \ldots, -1)$ on a path from $j$ to $s$ to the flow $F_1$:
$$
\begin{aligned}
\hat{F_1}(i_k,j_k) &= F_1(i_k, j_k) + \epsy_3, 1 \leq k \leq n-1
\\
\textrm{ and } \quad \hat{F_1}(i_{k-1},j_{k}) &= F_1(i_{k-1}, j_{k}) - \epsy_3, 2 \leq k \leq n. 
\end{aligned}
$$
For the other entries, $\hat{F_1}(e) = F_1(e)$. 
\item 
add  $\epsy_4\pcm  (1, -1, 1, \ldots, -1)$ on a path from $j$ to $s'$ to the flow $\hat{F_1}$: 
$$
\begin{aligned}
F_1''(i''_k,j''_k) &= \hat{F_1}(i''_k, j''_k) + \epsy_4\pcm , 1 \leq k \leq m-1
\\
\textrm{ and } \quad F_1''(i''_{k-1},j''_{k}) &= \hat{F_1}(i''_{k-1}, j''_{k}) - \epsy_4\pcm , 2 \leq k \leq m.
\end{aligned}
$$
For the other entries, $F_1'' (e) = \hat{F_1}(e)$. 
\end{itemize}
The rest of the proof is now similar. We define a new randomized policy using $F''$,
and   lemma holds for $k=j$, $m=j_0$, $\epsy_0 = \epsy_3 - \delta$, $\epsy'_0 = (c_j + c_{j_0}) (\epsy_3 - \delta) + c_{j_0} \epsy_4 \pcm > 0$,  %(for $\delta$ small enough), 
for $\delta$ small enough.
\end{proof}
%
%We get 
%$$
%\Expect\Bigl[ \sum_{k \; : \; q_k > 0} c_k Q_k (t+1) \mid Q(t)=q\Bigr] 
%%= \sum_{k \not \in \{i,j\}\; : \; x_k > 0} c_k \Expect\Bigl[ Q_k (t+1) \mid Q(t)=q\Bigr] 
%%+ c_i \Expect\Bigl[  Q_i (t+1) \mid Q(t)=x\Bigr]  + c_{s'} \Expect\Bigl[  Q_{s'} (t+1)\mid Q(t)=x\Bigr] 
%%\le  \left( \sum_{k \not \in \{i,j\}\; : \; q_k > 0}  c_k q_k \right)
%%			 + c_i (q_i - \epsy_1) + c_{j}(q_{j} - \epsy_1). 
%\le  \left( \sum_{k \; : \; q_k > 0}  c_k q_k \right)
%			 - (c_j - \barc_-)\epsy_3 - c_j \epsy_4\pcm . 
%$$ 

\end{document}